\documentclass[aps,prl,twocolumn,notitlepage,superscriptaddress,showpacs,nofootinbib]{revtex4-1}
\usepackage{amsmath}
\usepackage{amssymb}
\usepackage{amsfonts}
\usepackage{enumerate}
\usepackage{mathrsfs}
\usepackage{graphicx}
\usepackage{epstopdf}
\usepackage{enumerate}
\usepackage{changepage}
\usepackage{bm}
\usepackage{multirow}
\usepackage{lipsum}
\usepackage{booktabs}
\usepackage{tikz}

\newtheorem{theorem}{Theorem}
\newtheorem{definition}{Definition}

\newtheorem{lemma}{Lemma}
\newtheorem{proposition}{Proposition}
\newtheorem{conjecture}{Conjecture}
\newtheorem{example}{Example}
\newtheorem{corollary}{Corollary}

\def\bcj{\begin{conjecture}}
	\def\ecj{\end{conjecture}}
\def\bcr{\begin{corollary}}
	\def\ecr{\end{corollary}}
\def\bd{\begin{definition}}
	\def\ed{\end{definition}}
\def\bea{\begin{eqnarray}}
	\def\eea{\end{eqnarray}}
\def\bem{\begin{enumerate}}
	\def\eem{\end{enumerate}}
\def\bex{\begin{example}}
	\def\eex{\end{example}}
\def\bim{\begin{itemize}}
	\def\eim{\end{itemize}}
\def\bl{\begin{lemma}}
	\def\el{\end{lemma}}
\def\bma{\begin{bmatrix}}
	\def\ema{\end{bmatrix}}
\def\bpf{\begin{proof}}
	\def\epf{\end{proof}}
\def\bpp{\begin{proposition}}
	\def\epp{\end{proposition}}
\def\bqu{\begin{question}}
	\def\equ{\end{question}}
\def\br{\begin{remark}}
	\def\er{\end{remark}}
\def\bt{\begin{theorem}}
	\def\et{\end{theorem}}

\def\squareforqed{\hbox{\rlap{$\sqcap$}$\sqcup$}}
\def\qed{\ifmmode\squareforqed\else{\unskip\nobreak\hfil
		\penalty50\hskip1em\null\nobreak\hfil\squareforqed
		\parfillskip=0pt\finalhyphendemerits=0\endgraf}\fi}
\def\endenv{\ifmmode\;\else{\unskip\nobreak\hfil
		\penalty50\hskip1em\null\nobreak\hfil\;
		\parfillskip=0pt\finalhyphendemerits=0\endgraf}\fi}
\newenvironment{proof}{\noindent \textbf{{Proof.~} }}{\qed}
\def\Dbar{\leavevmode\lower.6ex\hbox to 0pt
	{\hskip-.23ex\accent"16\hss}D}
\makeatletter
\def\url@leostyle{%
	\@ifundefined{selectfont}{\def\UrlFont{\sf}}{\def\UrlFont{\small\ttfamily}}}
\makeatother
\def\bcj{\begin{conjecture}}
	\def\ecj{\end{conjecture}}
\def\bcr{\begin{corollary}}
	\def\ecr{\end{corollary}}
\def\bd{\begin{definition}}
	\def\ed{\end{definition}}
\def\bea{\begin{eqnarray}}
	\def\eea{\end{eqnarray}}
\def\bem{\begin{enumerate}}
	\def\eem{\end{enumerate}}
\def\bex{\begin{example}}
	\def\eex{\end{example}}
\def\bim{\begin{itemize}}
	\def\eim{\end{itemize}}
\def\bl{\begin{lemma}}
	\def\el{\end{lemma}}
\def\bpf{\begin{proof}}
	\def\epf{\end{proof}}
\def\bpp{\begin{proposition}}
	\def\epp{\end{proposition}}
\def\bqu{\begin{question}}
	\def\equ{\end{question}}
\def\br{\begin{remark}}
	\def\er{\end{remark}}
\def\bt{\begin{theorem}}
	\def\et{\end{theorem}}

\def\btb{\begin{tabular}}
	\def\etb{\end{tabular}}

	\newcommand{\nc}{\newcommand}
	
	\nc{\bbA}{\mathbb{A}} \nc{\bbB}{\mathbb{B}} \nc{\bbC}{\mathbb{C}}
	\nc{\bbD}{\mathbb{D}} \nc{\bbE}{\mathbb{E}} \nc{\bbF}{\mathbb{F}}
	\nc{\bbG}{\mathbb{G}} \nc{\bbH}{\mathbb{H}} \nc{\bbI}{\mathbb{I}}
	\nc{\bbJ}{\mathbb{J}} \nc{\bbK}{\mathbb{K}} \nc{\bbL}{\mathbb{L}}
	\nc{\bbM}{\mathbb{M}} \nc{\bbN}{\mathbb{N}} \nc{\bbO}{\mathbb{O}}
	\nc{\bbP}{\mathbb{P}} \nc{\bbQ}{\mathbb{Q}} \nc{\bbR}{\mathbb{R}}
	\nc{\bbS}{\mathbb{S}} \nc{\bbT}{\mathbb{T}} \nc{\bbU}{\mathbb{U}}
	\nc{\bbV}{\mathbb{V}} \nc{\bbW}{\mathbb{W}} \nc{\bbX}{\mathbb{X}}
	\nc{\bbZ}{\mathbb{Z}}
	
	\nc{\bA}{{\bf A}} \nc{\bB}{{\bf B}} \nc{\bC}{{\bf C}}
	\nc{\bD}{{\bf D}} \nc{\bE}{{\bf E}} \nc{\bF}{{\bf F}}
	\nc{\bG}{{\bf G}} \nc{\bH}{{\bf H}} \nc{\bI}{{\bf I}}
	\nc{\bJ}{{\bf J}} \nc{\bK}{{\bf K}} \nc{\bL}{{\bf L}}
	\nc{\bM}{{\bf M}} \nc{\bN}{{\bf N}} \nc{\bO}{{\bf O}}
	\nc{\bP}{{\bf P}} \nc{\bQ}{{\bf Q}} \nc{\bR}{{\bf R}}
	\nc{\bS}{{\bf S}} \nc{\bT}{{\bf T}} \nc{\bU}{{\bf U}}
	\nc{\bV}{{\bf V}} \nc{\bW}{{\bf W}} \nc{\bX}{{\bf X}}
	\nc{\ba}{{\bf a}} \nc{\be}{{\bf e}} \nc{\bu}{{\bf u}}
	\nc{\brr}{{\bf r}} \nc{\bx}{{\bf x}}
	
	\nc{\cA}{{\cal A}} \nc{\cB}{{\cal B}} \nc{\cC}{{\cal C}}
	\nc{\cD}{{\cal D}} \nc{\cE}{{\cal E}} \nc{\cF}{{\cal F}}
	\nc{\cG}{{\cal G}} \nc{\cH}{{\cal H}} \nc{\cI}{{\cal I}}
	\nc{\cJ}{{\cal J}} \nc{\cK}{{\cal K}} \nc{\cL}{{\cal L}}
	\nc{\cM}{{\cal M}} \nc{\cN}{{\cal N}} \nc{\cO}{{\cal O}}
	\nc{\cP}{{\cal P}} \nc{\cQ}{{\cal Q}} \nc{\cR}{{\cal R}}
	\nc{\cS}{{\cal S}} \nc{\cT}{{\cal T}} \nc{\cU}{{\cal U}}
	\nc{\cV}{{\cal V}} \nc{\cW}{{\cal W}} \nc{\cX}{{\cal X}}
	\nc{\cZ}{{\cal Z}}
	
	\nc{\hA}{{\hat{A}}} \nc{\hB}{{\hat{B}}} \nc{\hC}{{\hat{C}}}
	\nc{\hD}{{\hat{D}}} \nc{\hE}{{\hat{E}}} \nc{\hF}{{\hat{F}}}
	\nc{\hG}{{\hat{G}}} \nc{\hH}{{\hat{H}}} \nc{\hI}{{\hat{I}}}
	\nc{\hJ}{{\hat{J}}} \nc{\hK}{{\hat{K}}} \nc{\hL}{{\hat{L}}}
	\nc{\hM}{{\hat{M}}} \nc{\hN}{{\hat{N}}} \nc{\hO}{{\hat{O}}}
	\nc{\hP}{{\hat{P}}} \nc{\hR}{{\hat{R}}} \nc{\hS}{{\hat{S}}}
	\nc{\hT}{{\hat{T}}} \nc{\hU}{{\hat{U}}} \nc{\hV}{{\hat{V}}}
	\nc{\hW}{{\hat{W}}} \nc{\hX}{{\hat{X}}} \nc{\hZ}{{\hat{Z}}}
	
	\nc{\hn}{{\hat{n}}}
	
	\def\dim{\mathop{\rm Dim}}

	\def\ghz{\mathop{\rm GHZ}}

      \def\gap{\mathop{\rm gap}}

 	\def\swap{\mathop{\rm Swap}}
	
	\def\max{\mathop{\rm max}}
	\def\min{\mathop{\rm min}}

	\def\rank{\mathop{\rm rank}}
	
	\newcommand{\bra}[1]{\langle#1|}
	\newcommand{\ket}[1]{|#1\rangle}
	
	\newcommand{\ketbra}[2]{|#1\rangle\!\langle#2|}

	\newcommand{\fl}[2]{\left\lfloor\frac{#1}{#2}\right\rfloor}
	\newcommand{\fc}[2]{\left\lceil\frac{#1}{#2}\right\rceil}

	\usepackage[
	colorlinks,
	linkcolor = blue,
	citecolor = blue,
	urlcolor = blue]{hyperref}
	\def \qed {\hfill \vrule height7pt width 7pt depth 0pt}
	
	\newcounter{lastnote}

\begin{document}
		\title{
 Entanglement detection length of multipartite quantum states\\
  }
	\author{Fei Shi}
\affiliation{QICI Quantum Information and Computation Initiative, School of Computing and Data Science,
The University of Hong Kong, Pokfulam Road, Hong Kong}

\author{Lin Chen}
\email[]{linchen@buaa.edu.cn}
\affiliation{LMIB(Beihang University), Ministry of Education, and School of Mathematical Sciences, Beihang University, Beijing 100191, China}


\author{Giulio Chiribella}
\email[]{giulio@cs.hku.hk}
\affiliation{QICI Quantum Information and Computation Initiative,  School of Computing and Data Science,
The University of Hong Kong, Pokfulam Road, Hong Kong}	 
\affiliation{Department of Computer Science, Parks Road, Oxford, OX1 3QD, United Kingdom}	
\affiliation{Perimeter Institute for Theoretical Physics, Waterloo, Ontario N2L 2Y5, Canada}

\author{Qi Zhao}
\email[]{zhaoqi@cs.hku.hk}
\affiliation{QICI Quantum Information and Computation Initiative, School of Computing and Data Science,
The University of Hong Kong, Pokfulam Road, Hong Kong}

\begin{abstract}
Multipartite entanglement is a valuable resource for quantum technologies.  
However, detecting this resource can be challenging: for genuine multipartite entanglement, the detection  may require global measurements that are hard to implement experimentally.  
Here we introduce the concept of  entanglement detection length, defined as the minimum number of particles that have to be jointly  measured in order to detect genuine multipartite entanglement. 
{For  symmetric states, we show that the entanglement detection length   can be determined by testing separability of the marginal states. 
For general states, we   provide an  upper bound  on the entanglement detection length based on  semidefinite programming.}  
We show that the entanglement detection length is generally smaller than the minimum observable length needed to uniquely determine a multipartite  state, and we provide examples achieving the maximum gap between these two quantities.
\end{abstract}
\maketitle



\textit{\textbf{Introduction.}}\textbf{---} Genuine multipartite entanglement (GME), a global form of entanglement arising in multipartite systems, is a powerful resource in many applications, including quantum communication \cite{karlsson1998quantum}, quantum computing \cite{steane1998quantum} and quantum metrology \cite{chin2012quantum}. Genuinely entangled states of up to 51  qubits have been generated in superconducting qubit systems \cite{cao2023generation}, and even larger sizes are likely to become accessible in the near future.

Detecting GME in the laboratory, however, is a challenging problem.  
 Many detection schemes require global measurements that are difficult to implement experimentally.
{To address this problem,  recent works developed searched for alternative schemes using  properties of the marginal states \cite{navascues2021entanglement,tabia2022entanglement,toth2005entanglement,jungnitsch2011taming,sperling2013multipartite,chen2014role,miklin2016multiparticle,paraschiv2018proving} or few-body observables \cite{toth2007optimal,guhne2009entanglement,toth2009spin,gittsovich2010multiparticle,bancal2014device,liang2015family,tura2014detecting,frerot2021optimal,frerot2022unveiling}. The detection of multipartite entanglement is also closely related to foundational problems, such  as  the  entanglement marginal problem \cite{navascues2021entanglement} and the entanglement transitivity problem \cite{tabia2022entanglement}.}

Despite many advances in related areas,
the problem of detecting GME 
through few-body observables 
still lacks a systematic treatment, and some of the most basic questions have remained unanswered so far. For example,  what is the minimum number of subsystems that have to be measured jointly in order to detect the GME of a given state?    And how many minimum-sized marginals have to be characterized?

In this paper,  we introduce  the { {\em entanglement detection length (EDL)}} of a genuinely entangled state, defined as the minimum integer $k$ such that the state's GME  can be detected by $k$-body observables, and as illustrated Fig.~\ref{fig:EDL}.  We show that EDL  is  relevant to the design of entanglement detection protocols, and provides a lens for analyzing various features of multipartite entangled states studied in the literature, such as GHZ states \cite{greenberger1989going}, $k$-uniform states \cite{scott2004multipartite}, Dicke states \cite{dicke1954coherence}, and graph states  \cite{hein2004multiparty}. We focus in particular on the case of symmetric states  \cite{i2016characterizing}, which play a key role in  quantum metrology  \cite{toth2014quantum,pezze2018quantum,oszmaniec2016random},  quantum nonlocality \cite{wang2012nonlocality,quesada2017entanglement},  quantum entanglement 
\cite{toth2009entanglement,eckert2002quantum,stockton2003characterizing,ichikawa2008exchange,tura2012four,augusiak2012entangled,wolfe2014certifying,yu2016separability,tura2018separability}.   {For symmetric states, we prove that the EDL can be determined by testing whether the $k$-body marginals are entangled, and we use this result to find the minimum number of minimum-sized marginals needed to detect GME.}  For more general states, we provide an upper bound on the EDL based on semidefinite programming.   

\begin{figure}[t]
    \centering
  \includegraphics[scale=0.38]{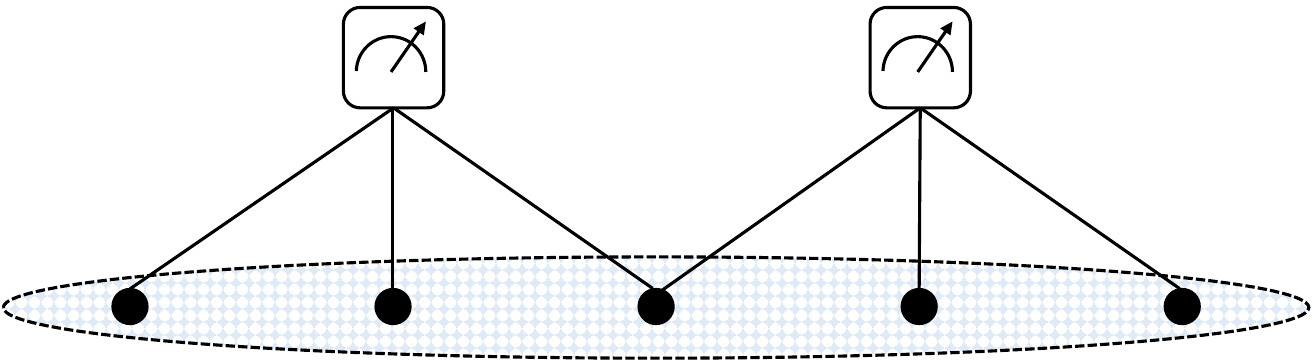}
		\caption{   {\bf Entanglement detection length.} The EDL is the minimum $k$ such that the genuine multipartite entanglement of a given multipartite state can be detected by $k$-body measurements.  In the figure, two $3$-body measurements detect  the GME of a $5$-partite system. 
  } \label{fig:EDL}
\end{figure}

We also study an analogue of the EDL,  called {{\em state determination length (SDL)}}, and defined  as the  minimum  integer $k$ such that the global state can be determined by $k$-body observables. This quantity is relevant to the  problem of reconstructing a global quantum state using experimental data generated from its marginals \cite{linden2002almost,linden2002parts,diosi2004three,cramer2010efficient,chen2012comment,chen2012ground,chen2013uniqueness,xin2017quantum}. Since determining a quantum state includes, in particular, determining whether the state is genuinely entangled or not,  the SDL is  an upper bound to the EDL. However, we show that the bound is strict, and provide examples of pure and mixed states that exhibit the maximum possible gap between the SDL and EDL.

Based on the EDL,   $n$-qubit genuinely entangled  states  can be divided into $n-1$ classes, as illustrated in  Fig.~\ref{fig:detection}.  Notably, this classification  is different from the classification
of genuinely entangled states based on stochastic local operations and classical communication (SLOCC)  \cite{eisert2001schmidt,dur2000three}, which for $n\geq 4$ gives rise to infinitely many inequivalent classes \cite{verstraete2002four,dietrich2022classification}. 
 Finally, we establish a connection between the EDL and the entanglement transitivity problem. Our results provide new insights into the structure of  multipartite entanglement, and can be used to design entanglement detection schemes that use the minimum number of experimental settings.



\begin{figure}[t]
    \centering
  \includegraphics[scale=0.45]{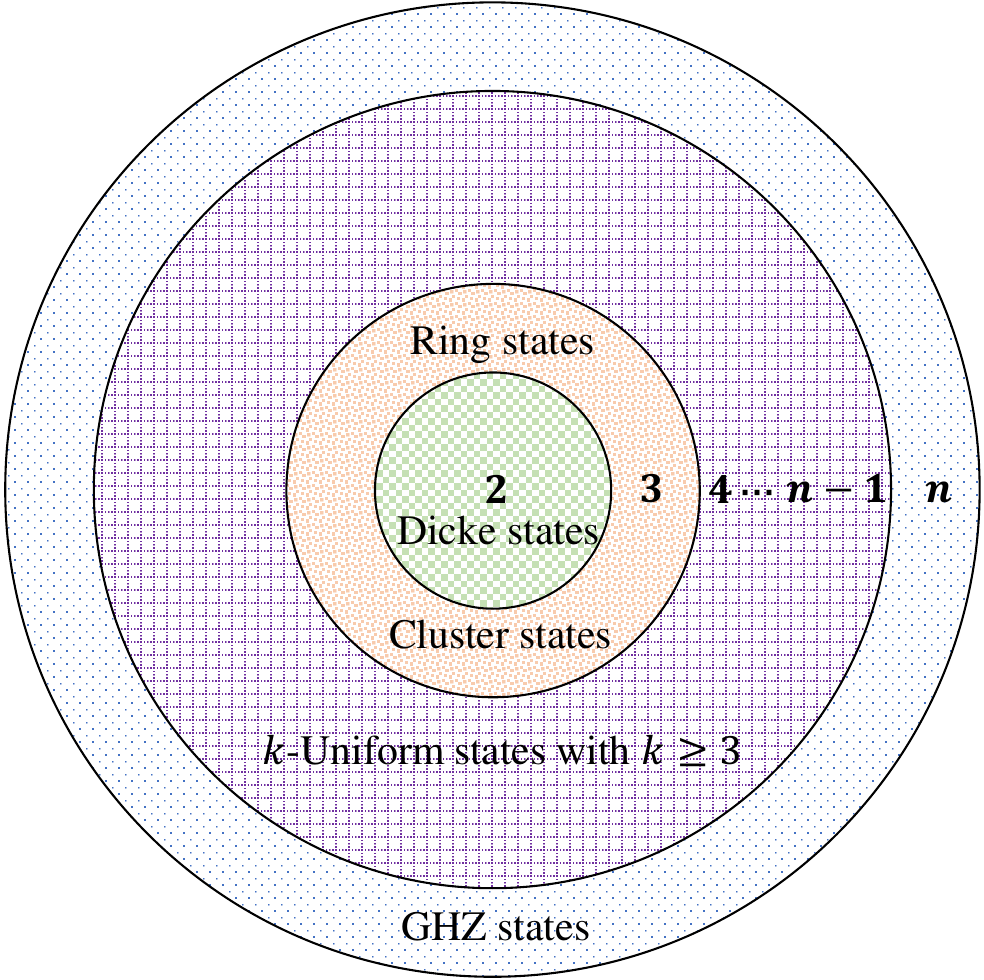}
		\caption{   {\bf   Classification of genuinely entangled states based on the EDL.} Genuinely entangled states can be divided into $n-1$ classes based on the values of their EDL.
  } \label{fig:detection}
\end{figure}

\smallskip

\textit{\textbf{EDL and SDL.}}\textbf{---} Consider a quantum system consisting of $n$  particles, labeled by integers in the finite set $[n]  : =  \{1,2,\dots,  n\}$. We denote by $\cH_i$ the Hilbert space of the $i$-th particle, and by  $\cH_{[n]}:=\otimes_{i=1}^n\cH_{i}$ the Hilbert space of the total system. 
The set of all density matrices on $\cH_{[n]}$ will be denoted by $\cD_{[n]}$. Subsystems will be identified by subsets of $[n]$. For a given subset $S  \subseteq [n]$, the Hilbert space of the corresponding subsystem is $\cH_S  :  = \bigotimes_{j \in  S} \,  \cH_j$ and the set of density matrices of $\cH_S$ is denoted by $\cD_S$.    For a given quantum state $\rho  \in  \cD_{[n]}$,   we will denote by $\rho_S \in \cD_S$  the marginal state of the state $\rho$ on system $S$, namely $\rho_S  :  =  {\sf Tr}_{\overline S}  [\rho]$, where ${\sf Tr}_{\overline S}$ denotes the partial trace over the Hilbert space $\cH_{\overline S}$ associated to the complementary subsystem $\overline S : = [n] \setminus S$. 

{A pure state $\ket{\psi}\in \cH$  is  \emph{separable with respect to the  bipartition $S|\overline{S}$} if  $\ket{\psi}=\ket{\alpha}_S\otimes\ket{\beta}_{\overline{S}}$, where $\ket{\alpha}_{S}$   and $\ket{\beta}_{\overline{S}}$ are pure states in $\cH_S$ and $\cH_{\overline S}$, respectively.   The state $\ket{\psi}$ is called {\em biseparable} if there exists a proper subset $ \emptyset  \not  =  S  \subsetneq [n]$ such that $\ket{\psi}$ is separable with respect to the bipartition $S|\overline{S}$.   The state $\ket{\psi}$ is called {\em fully separable} if it is separable with respect to all possible bipartitions. A mixed  state $\rho\in \cD_{[n]}$ is biseparable (respectively, fully separable) if it can be written as a convex combination of biseparable (respectively, fully separable) pure states.
A state $\rho\in \cD_{[n]}$ is  \emph{entangled} if it is not fully separable, and   {\em genuinely entangled} if it is not biseparable.}    

We now consider the problem of determining the entanglement of a state from the specification of its marginals.  
Let $\cS   :=  \{S_1,  \dots,  S_t\}$ be a collection of subsets of $[n]$.    We define the {\em compatibility set}  $\cC(\rho  , \cS)$  as the set of all density matrices 
$\sigma \in \cD_{[n]}$ that have the same marginals as $\rho$ on the   subsystems labeled by elements of $\cS$:  in formula, 
\begin{align}\label{compatibilityset}
\cC(\rho, \cS ):=\big\{\sigma\in\cD_{[n]} ~\big|~ \sigma_{S}= \rho_{S}, \ \forall \,  S\in \cS \big\}\,.
\end{align}

 Note that the  compatibility
set is a convex set \cite{chen2012ground}.   
If the  set  $\cC(\rho, \cS )$ contains  only the state $\rho$, we say that   $\cS$ {\em determines  $\rho$}.     
If the compatibility set  $\cC(\rho, \cS )$ contains  only genuinely entangled states,  we say that  $\cS$ {\em detects  $\rho$'s GME}.    The set of all  $\cS$ that determine $\rho$ will be denoted by  ${\sf SD}  (\rho)$, while the set of all $\cS$ that  detect $\rho$'s GME will be denoted by  ${\sf ED}  (\rho)$. For a genuinely entangled state $\rho$, one has the inclusion ${\sf SD}  (\rho)  \subseteq {\sf ED}  (\rho)$.

The   {\em entanglement detection length  (EDL)} and the {\em state determination length (SDL)} of the state $\rho$ are  defined as 
\begin{align}\label{EDL}
l(\rho)  :  =\min_{\cS   \in {\sf ED}  (\rho) } \max_{S\in  \cS}   \, |S| \, ,
\end{align}
and 
\begin{align}\label{SDL}
L(\rho)  :  =\min_{\cS   \in {\sf SD}  (\rho) } \max_{S\in  \cS}   \, |S| \, ,
\end{align}
respectively.  For pure states $|\psi\rangle$, we will  slightly abuse the notation, writing $l (|\psi\rangle)$ and  $L (|\psi\rangle)$ in place of $l (|\psi\rangle\langle \psi|)$ and $L(|\psi\rangle\langle \psi|)$, respectively.  

The SDL (respectively, EDL) quantifies the minimum    number of particles that have to be jointly measured  in order to determine a state (respectively, detect its GME).    In particular,  every experimental scheme that determines the state $\rho$ (respectively, certifies the presence of GME in the state $\rho$)   must include at least one $k$-body observable with $k\geq L(\rho)$ (respectively, $k\geq l(\rho)$).  {In   Section I of the Supplemental Material (SM), we show  that the EDL and SDL  are both  invariant under local unitary  (LU) operations, but may change under SLOCC.  Moreover, since they are both integer-valued, they are nonconvex.} 


 The definitions of  EDL and SDL imply the inequality $l(\rho) \le L(\rho)$ for every genuinely entangled state $\rho$. {This inequality is generally strict: for example, a 4-qubit pure state with  non-zero gap between the EDL and SDL  can be found, {\em e.g.,} in Ref.~\cite{miklin2016multiparticle} (see  Section~IX of the SM for details).   Later in the paper, we will provide examples of states with maximum gap between the EDL and SDL.}

\smallskip  
{\em \textbf{Examples.}}\textbf{---} The notions of SDL and EDL are related to  several important features  of the key examples of multipartite entangled states considered  in the literature. For example,  $n$-qubit GHZ states \cite{greenberger1989going} have maximum SDL, equal to $n$,  since all their $k$-particle marginals with $k\le n-1$ are compatible with the separable state $\sigma  = 1/2 (|0\rangle \langle 0|)^{\otimes n} +    1/2 (|1\rangle \langle 1|)^{\otimes n}$. For the same reason, they also have maximum EDL.  In fact, one can show that   the SDL (respectively, EDL) of an $n$-qubit (respectively, genuinely entangled) pure state  is equal to $n$ if and only if the state is  ``$\ghz$-like'' \cite{walck2009only}, that is, it is LU-equivalent to an $n$-qubit generalized $\ghz$ state of the form $\alpha  \,  |0\rangle^{\otimes n}  +  \beta \, |1\rangle^{\otimes n}$, with $\alpha\beta\neq 0$.   
Lower bounds on the EDL and SDL are also immediate for  $k$-uniform states \cite{scott2004multipartite,goyeneche2014genuinely}, {\em i.e.,} $n$-qubit pure states  for which all the   $k$-body marginals are maximally mixed: since the $k$-body marginals are compatible with the maximally mixed $n$-qubit state,   the SDL and  EDL  $k$-uniform states must be strictly larger than $k$.


{Examples of entangled states with minimum SDL/EDL  are provided unique ground states of 2-local Hamiltonians, such as projected pair entangled states \cite{10.5555/2016976.2016982}: when entangled, they have   SDL equal to 2, since their $2$-body marginals determine the expectation value of the Hamiltonian.  When  genuinely entangled, they have  EDL equal to 2. Dicke states, {\em i.e.,} $n$-qubit states of the form 
\begin{equation}
 \ket{D_n^i}:=\frac{1}{\sqrt{\binom{n}{i}}}\sum_{\begin{subarray}{c}
   \\   s_j   \in  \{0,1\}, \, \forall \,  j\in [n] \\  \\
     \sum_{j=1}^n  s_j=i 
     \end{subarray}}  |s_1\rangle \otimes |s_2\rangle \otimes \cdots\otimes | s_n\rangle  \, ,
\end{equation}  
are an example of this situation ~\cite{chen2012correlations}. For $i\in  [1,n-1]$, the Dicke state $|D_n^i\rangle$ is genuinely entangled \cite{dicke1954coherence,lucke2014detecting}, and its EDL is $l(|D_n^i\rangle)   =  2$. More generally, unique ground states of $k$-local Hamiltonians have SDL  upper bounded by $k$, and the same bound applies to the EDL if they are genuinely entangled \cite{yu2023learning}. }


Graph states \cite{hein2004multiparty,guhne2005bell,hein2006entanglement},  including in particular cluster states \cite{briegel2001persistent} and ring states \cite{baccari2020scalable}, are another interesting class of multipartite states.    For cluster states and ring states, we show that the SDL and EDL coincide and are equal to 3.   For more general genuinely entangled graph states,  we show that the SDL and EDL are between $3$ and 1 plus the maximum degree of a vertex in the corresponding graph (see Section II  of the SM for  details).

\smallskip 

{\em \textbf{EDL and SDL of symmetric states.}}\textbf{---} We now provide a complete characterization of the EDL for $n$-qubit symmetric states \cite{i2016characterizing}, that is, states with support in the completely symmetric subspace  (see Section~III of the SM for a more detailed introduction).    Notably, a symmetric  state is  either fully separable  or  genuinely entangled \cite{ichikawa2008exchange}.

For  entangled symmetric states, we have the following:
\begin{proposition}\label{prop:symmetric} 
For an $n$-qubit entangled symmetric state $\rho$,   $l(\rho)=  \min  \{  k\mid   \rho_{[k]} \text{ is entangled} \}$.  
\end{proposition}
The proof is provided in Section~IV of the SM, where we compute the EDL in several special cases of  pure and mixed symmetric states.   

{Proposition \ref{prop:symmetric}    implies that every symmetric state with EDL equal to 2 is  a robust multipartite entangled state \cite{luo2021robust}, meaning that all its $k$-body marginals with $k\ge 2$  are genuinely entangled.}   

Proposition \ref{prop:symmetric}  reduces the calculation of the EDL to the problem of determining whether a symmetric state is entangled or not. This reduction is important, because the entanglement of symmetric states is well studied \cite{toth2009entanglement,eckert2002quantum,stockton2003characterizing,ichikawa2008exchange,tura2012four,augusiak2012entangled,wolfe2014certifying,yu2016separability,i2016characterizing,tura2018separability,imai2024collective}. For example, all symmetric states that are diagonal in the basis of the Dicke states are separable if and only if their density matrix has positive partial transpose with respect to $\lfloor n/2\rfloor$ qubits \cite{yu2016separability}.

An explicit  expression for the SDL can be provided  for a family of Dicke-diagonal symmetric states: 
\begin{proposition}\label{pro:symmetric_SDL}  Let  $\rho$ be a symmetric state of the diagonal  form   $\rho=\sum_{i=0}^k\lambda_i\ket{D_n^i}\bra{D_n^i}$ for some integer $k\le n$, where    $\lambda_i$ is nonzero for all but one value of $i$, namely $\lambda_i\not  =  0,  \, \forall \, i\in  \{0,\dots, k\} \setminus \{i_*\}$ for some $i_*\in  \{0,\dots,  k\}$.
  Then,  $L(\rho)=k   + \min  \{[(k-i_*)\mod 2],  |n-k|\}$.  
  \end{proposition}

The proof of Proposition~\ref{pro:symmetric_SDL} is given in Section~V of the SM, where we compute the SDL for several examples of diagonal symmetric states.

For symmetric states, the collection of marginals needed to determine a state (detect its GME)  have a simple graph-theoretic characterization. {A hypergraph $G= (V,E)$ consists of a set of vertices $V$ and a set of hyperedges $E$,  a hyperedge being a  nonempty subset of $V$. 
In this language, every collection $\cS$ of subsets of  $[n]$ defines a hypergraph $G=  ([n], \cS)$. A hypergraph is connected if, for any partition of its vertex set into two non-empty sets $X$ and $Y$, there exists at least one hyperedge  connecting a vertex in $X$ to a vertex in $Y$.}   

\begin{proposition}\label{prop:graphtheory}
For an $n$-qubit entangled symmetric state $\rho$, a  collection $\cS$ determines $\rho$ (detects $\rho$'s GME) if and only if   the hypergraph $G=  ([n],\cS)$ is connected and $|S| \ge L(\rho)$ ($|S| \ge l(\rho)$)  $\exists \, S\in \cS$.   
\end{proposition}


This result, provided in Section VI   of the SM, can be used to compute the minimum number of  marginals needed to determine a state/guarantee GME.    Without loss of generality, we can restrict the attention to $L(\rho)$-body  ($l(\rho)$-body) marginals for determining (detecting the GME of) the state $\rho$.  {We denote by $\cS_k$  the collection of all $k$-subsets of $[n]$, that is $|\cS_k|=\binom{n}{k}$, and $|S| = k$ for every $S\in  \cS_k$.} The minimum number of  marginals  needed to determine  $\rho$ 
is
\begin{equation}\label{Mrho}
    M (\rho):= \min_{\cS \in  {\sf SD}  (\rho) }  \left\{ |\cS|  ~|~  \cS  \subseteq  \cS_{L(\rho)}  \right\} \, , 
\end{equation}
while the minimum number  marginals needed to detect $\rho$'s GME is 
\begin{equation}\label{mrho}
    m (\rho):= \min_{\cS \in  {\sf ED}  (\rho) }    \left\{ |\cS|~|~       \cS  \subseteq  \cS_{l (\rho)} \right\}
       \, . 
\end{equation}
For an $n$-qubit entangled symmetric state $\rho$, Proposition \ref{prop:graphtheory} yields the simple expressions  $M(\rho)   =\lceil  (n-1)/(L(\rho)-1) \rceil $ and  $m(\rho)   =\lceil  (n-1)/(l(\rho)-1) \rceil$  (see  Section VI of the SM for the derivation).

The graph-theoretic characterization of  Proposition \ref{prop:graphtheory} has also an application to  the problem of  entanglement transitivity \cite{tabia2022entanglement}.   A set of marginals $\{\rho_S\mid S\in \cS \}$  exhibits \emph{entanglement transitivity} in a target system $T\notin \cS$  if for all global states $\rho$ compatible with these marginals, the marginal $\rho_T$ is entangled. By  definition of the EDL, 
the set of all $l(\rho)$-body marginals $\{\rho_S\mid S\in \cS_{l(\rho)} \}$   exhibits entanglement transitivity in $[n]$, for every   genuinely entangled state $\rho$.  In the special case of symmetric states, one can prove a  stronger result: 


\begin{proposition}\label{prop:transitivity}
   For an $n$-qubit  entangled symmetric state $\rho$, if the hypergraph $G=  ([n],\cS)$ is connected and $|S| \ge l(\rho)$  $\exists \, S\in \cS$, then  the set of marginals $\{\rho_S\mid S\in \cS \}$ exhibits entanglement transitivity in any $T$ with $|T|\geq l(\rho)$.  
\end{proposition}

The proof of Proposition~\ref{prop:transitivity} is given in Section~VII of the SM,  where we also present examples.  




\smallskip 
 
{\em \textbf{General upper bound on  the EDL.}}\textbf{---} We now provide a general  upper bound on the EDL of arbitrary genuinely entangled states. The bound is based on the notion of fully decomposable witness (FDW) \cite{jungnitsch2011taming}: an $n$-qubit observable  $W$ is an FDW   if for every bipartition $S|\overline{S}$, there exist positive semidefinite operators $P_S$ and $Q_S$ such that $W=P_S+Q_S^{T_S}$, where $T_S$ is the partial transpose over $\cH_S$.   This definition implies that the expectation value of $W$ is non-negative for every biseparable state. Hence, a negative value of ${\sf Tr} [W\rho]$ implies  that $\rho$ is genuinely entangled.



We now restrict the attention to FDWs that can be constructed from measurements on a given collection of subsystems, labeled by a collection  $\cS$ of subsets of $[n]$, that is, FDWs of the form  $W =  \sum_{S  \in \cS}  H^S  \otimes I^{\overline S} $, where $H^S$ is an Hermitian operator on $\cH_S$ and $I^{\overline{S}}$ is the identity operator on  $\cH_{\overline S}$.  For a given state $\rho$, minimizing the expectation value  ${\sf Tr} [W \rho]$ over all FDWs of this form is an SDP, whose optimal value is  
\begin{equation}\label{eq:sdp}
\begin{aligned}
\alpha  (\rho,  \cS) := \min \quad &{\sf Tr}(W\rho)\\
\text{s.t.} \quad &{\sf Tr}(W)=1, \\
&W=\sum_{S\in \cS} H^{S}\otimes \bbI^{\overline{S}},\\
&\text{$W$ is fully decomposable.}
\end{aligned}
\end{equation}
For any given collection $\cS$, Eq.~\eqref{eq:sdp} can be solved by using the convex optimization package in Matlab \cite{cvx}.  We then have the following: 
\begin{proposition}\label{prop:sdp}
For every state $\rho$, one has    $l(\rho)   \le  \min\{ k~|~  \alpha  (\rho,  \cS_k)  <0\}$. 
\end{proposition}
We also observe that evaluation of $l(\rho)$ from the SDP is robust to small amounts of white noise: for a state of the form  $\rho_p  =(1-p)\, \rho    +  p  \, I/2^n$, we have $l(\rho_p)\leq k$ for all  $p\in [0,\frac{2^n\alpha (\rho,  \cS_k)}{2^n\alpha  (\rho ,\cS_k)-1})$ when $\alpha (\rho,  \cS_k)<0$.  Moreover, we can also evaluate the SDL of pure states from SDP. See Section~VIII of the SM for all details, where we also show some examples.  


\smallskip

{\em \textbf{Maximum gap between EDL and SDL.}}\textbf{---}  We conclude the paper by  providing examples of states for which the gap   is maximum.  

Let us define 
$\gap  (\rho):  =  L(\rho)  -  l(\rho)$ for a genuinely entangled state $\rho$.   Since for every genuinely entangled state $\rho$, one has the trivial bounds $L(\rho) \le n$  and $l(\rho)  \ge 2$, the maximum value of the gap is at most $n-2$. On the other hand, this maximum value cannot be obtained for genuinely entangled pure states, because the SDL of an $n$-qubit pure state is $n$ only if the state is $\ghz$-like, which in turn implies that the EDL is $n$.   
Hence, the gap between the SDL and EDL is at most $n-3$ for pure states.   We show that this value can be achieved by suitable superpositions of Dicke states: 
\begin{proposition}\label{prop:pure}
For an $n$-qubit genuinely entangled pure state $|\psi\rangle$, one has $\gap ( |\psi\rangle  )  \le \max\{0,n-3\} $.   The bound  is attained by  every  state of the form  $\ket{\psi}=\alpha  \ket{D_n^1}+  \beta \, \ket{D_n^n}$, where $\alpha$ and $\beta$ are two amplitudes satisfying the condition     $\frac{n^2-2n}{n^2-2n+1}<|\alpha|^2<1$.
\end{proposition}

For mixed states, the maximum value can  attained by suitable  mixtures of Dicke states:

\begin{proposition}\label{prop:mixed}
For an $n$-qubit  genuinely entangled mixed state $\rho$,  one has   $\gap  (\rho)\le \max\{0, n-2\}$.    The bound is attained by  every  state of the form   
$\rho=\sum_{i=0}^{n}\lambda_i\ketbra{D_n^i}{D_n^i}$ where $(\lambda_i)_{i=0}^n$ are probabilities satisfying the conditions $\lambda_0\lambda_n\neq 0$ and $\sum_{i,j=0}^n \,  (n-i)j \,\lambda_i\lambda_j  [  (n-i-1)  (j-1)-i(n-j)] < 0$. 
\end{proposition}

An example of mixed state with maximum gap is   $\rho=\frac{1}{24}\ketbra{D_4^0}{D_4^0}+\frac{1}{3}\ketbra{D_4^1}{D_4^1}+\frac{1}{2}\ketbra{D_4^2}{D_4^2}+\frac{1}{12}\ketbra{D_4^3}{D_4^3}+\frac{1}{24}\ketbra{D_4^4}{D_4^4}$. The proof of Propositions~\ref{prop:pure} and \ref{prop:mixed} can be found in Section~IX of the SM.

\smallskip

\textit{\textbf{Conclusions.}}\textbf{---} We have introduced the notion of EDL, determining the minimum number of particles that have to be jointly measured in order to detect genuine multipartite entanglement. This quantity 
can be used to analyze different aspects of multipartite entanglement, including entanglement transitivity and robustness.    


A promising avenue of future research is to extend the notion of EDL from GME to other properties of multipartite quantum states. {For example, it is interesting to consider the minimum number of particles needed to detect entanglement depth \cite{guhne2009entanglement,liang2015family,tura2019optimization,lin2019exploring,aloy2019device,lu2018entanglement} and GME without multipartite correlations \cite{schwemmer2015genuine, klobus2019higher}. Similarly, one could explore  the notion of nonlocality length, defined as the minimum length of the Bell correlations needed to detect nonlocality in a multipartite quantum state \cite{wurflinger2012nonlocal,vertesi2014certifying}.}



\section*{Acknowledgments}
We thank Zhuo Chen, Xiao Yuan, Xiongfeng Ma, You Zhou, Zhenhuan Liu, Jun Gao and Xiande Zhang  for their helpful discussion and suggestions. F.~S.~ and G.~C.~acknowledges funding from
 the Hong Kong Research Grant Council through Grants No.
 17300918 and No. 17307520, and through the Senior Research
 Fellowship Scheme SRFS2021-7S02. 
 This publication was made possible through the support of the ID\# 62312 grant
 from the John Templeton Foundation, as part of the ‘The Quantum Information Structure of Spacetime’
   Project (QISS). The opinions expressed in this project are those of the
 authors and do not necessarily reflect the views of the John Templeton Foundation. Research at the Perimeter Institute is supported by the
 Government of Canada through the Department of Innovation, Science and Economic Development Canada and by the
 Province of Ontario through the Ministry of Research, Innovation and Science. L.~C.~was supported by the NNSF of China (Grant No. 12471427), and the Fundamental Research Funds for the Central Universities (Grant Nos. KG12040501,
 ZG216S1810 and ZG226S18C1). Q.Z. acknowledges funding from HKU Seed Fund for Basic
 Research for New Staff via Project 2201100596, Guangdong Natural Science Fund—General Programme via Project 2023A1515012185, National Natural Science Foundation of China (NSFC) Young Scientists Fund via Project 12305030, 27300823, Hong Kong Research Grant Council (RGC) via No. 27300823, and NSFC/RGC Joint Research Scheme via Project N\_HKU718/23.




\bibliographystyle{IEEEtran}
\bibliography{reference}

 \newpage

 \begin{widetext}
   \appendix  
\section{\begin{large}
    Supplemental Material: Entanglement detection length of multipartite quantum states
\end{large}}
\begin{sloppypar}
We  denote by $\cS_k$  the collection of all $k$-subsets of $[n]$, that is $|\cS_k|=\binom{n}{k}$, and $|S| = k$ for every $S\in  \cS_k$.

A hypergraph $G= (V,E)$ consists of a set of vertices $V$ and a set of hyperedges $E$,  a hyperedge being a  nonempty subset of $V$. The hypergraph is $k$-uniform if every hyperedge has cardinality $k$.    In this language, every collection $\cS$ of subsets of  $[n]$ defines a hypergraph $G=  ([n], \cS)$. A hypergraph  is connected if, for any partition of its vertex set into two non-empty sets $X$ and $Y$, there exists at least one hyperedge  connecting a vertex in $X$ to a vertex in $Y$.

 \setcounter{proposition}{0}

 

For a state $\rho$, we denote by $\cR(\rho)$  the range of $\rho$.  We denote by $\cH_i$ the Hilbert space of the $i$-th particle, and by  $\cH_{[n]}:=\otimes_{i=1}^n\cH_{i}$ the Hilbert space of the total system. 
 Subsystems will be identified by subsets of $[n]$. For a given subset $S  \subseteq [n]$, the Hilbert space of the corresponding subsystem is $\cH_S  :  = \bigotimes_{j \in  S} \,  \cH_j$.

We denote by  $\rho_{\{i_1,i_2,\ldots, i_k\}}$ (respectively, $\ketbra{\psi}{\psi}_{\{i_1,i_2,\ldots, i_k\}}$)  the marginal of $\rho$ (respectively,  $\ket{\psi}$) on the subsystems $\{i_1,i_2,\ldots, i_k\}$, and by $\rho_{[k]}$ (respectively, $\ketbra{\psi}{\psi}_{[k]}$)  the marginal of $\rho$ (respectively, $\ket{\psi}$) on the subsystems $\{1,2,\ldots,k\}$.

 Based on the definitions of EDL and SDL, we know that for a state $\rho$ (respectively, genuinely entangled $\rho$), $L(\rho)\geq k$ (respectively, $l(\rho)\geq k$) if and only if $\cC(\rho,\cS_{k-1})$ contains a state $\sigma\neq \rho$ (respectively, contains at least one biseparable state), and $L(\rho)\leq k$ (respectively, $l(\rho)\leq k$) if and only if $\cC(\rho,\cS_{k})$ contains only $\rho$ (respectively, contains only genuinely entangled states).

 \section{I. \ Some properties for EDL and SDL}

\begin{lemma}\label{lemma:L_LU}
For an $n$-qubit state $\rho$ (respectively, genuinely entangled $\rho$),  $L(\rho)=L(U_1\otimes\cdots \otimes U_n\rho  U_1^{\dagger}\otimes \cdots\otimes U_n^{\dagger})$ (respectively, $l(\rho)=l(U_1\otimes\cdots \otimes U_n\rho  U_1^{\dagger}\otimes \cdots\otimes U_n^{\dagger})$) for any unitary $U_i$, $1\leq i\leq n$.

\end{lemma}
\begin{proof}
    Assume $L(\rho)=k$ (respectively, $l(\rho)=k$), then there exists a state $\sigma\neq \rho$ (respectively, a biseparable state  $\sigma$) such that $\sigma\in \cC(\rho,\cS_{k-1})$, and $ \cC(\rho,\cS_{k})$ contains only $\rho$ (respectively, genuinely entangled states).
    
Let $\rho'=(\otimes_{i=1}^nU_i)\rho (\otimes_{i=1}^nU_i^{\dagger})$, and $\sigma'=(\otimes_{i=1}^nU_i)\sigma (\otimes_{i=1}^nU_i^{\dagger})$.  For every $S\in \cS_{k-1}$, we have $\sigma'_S=(\otimes_{i\in S}U_i)\sigma_S (\otimes_{i\in S}U_i^{\dagger})=(\otimes_{i\in S}U_i)\rho_S (\otimes_{i\in S}U_i^{\dagger})=\rho'_S$, which means $\sigma'\in \cC(\rho',\cS_{k-1})$.  Since $\sigma\neq \rho$ (respectively, $\sigma$ is  biseparable), we obtain that $\sigma'\neq \rho'$ (respectively, $\sigma'$ is also biseparable). Thus $L(\rho')\geq k$ (respectively, $l(\rho')\geq k$).  

Assume $\tau' \in \cC(\rho',\cS_{k})$. We denote $\tau=(\otimes_{i=1}^nU_i^{\dagger})\tau '(\otimes_{i=1}^nU_i)$.  
For every $S\in \cS_{k}$, we have $\tau'_S=(\otimes_{i\in S}U_i)\tau_S (\otimes_{i\in S}U_i^{\dagger})=(\otimes_{i\in S}U_i)\rho_S (\otimes_{i\in S}U_i^{\dagger})=\rho_S'$. Then we have  $\tau_S=\rho_S$ for every $S\in \cS_k$, i.e. $\tau\in \cC(\rho, \cS_k)$. Since $\cC(\rho,\cS_{k})$ contains only $\rho$ (respectively, genuinely entangled states), we obtain that  $\tau=\rho$ (respectively, $\tau$ is genuinely entangled), which implies that $\tau'=\rho'$ (respectively, $\tau'$ is genuinely entangled). Thus $\cC(\rho',\cS_{k})$ contains only $\rho'$ (respectively, genuinely entangled states)
and  $L(\rho')\leq k$ (respectively, $l(\rho')\leq k$).

 Above all, we have $L(\rho')= k$ (respectively, $l(\rho')=k$).   
\end{proof}
\vspace{0.4cm}




However, the SDL (respectively, EDL) may change under SLOCC. For example, consdier $\ket{{\ghz}_3}=\frac{1}{\sqrt{2}}(\ket{000}+\ket{111})$. Let $A=\begin{pmatrix}
    1 & \frac{\sqrt{2}}{\sqrt{3}} \\
    0 &\frac{1}{\sqrt{3}}
\end{pmatrix}$, 
then $\ket{\psi}=A\otimes I\otimes I\ket{{\ghz}_3}=\frac{1}{\sqrt{2}}\ket{000}+\frac{1}{\sqrt{3}}\ket{011}+\frac{1}{\sqrt{6}}\ket{111}$, which is not LU-equivalent to a $3$-qubit generalized $\ghz$  state $\alpha \ket{000}+  \beta \, \ket{111}$ with $\alpha\beta\neq 0$. Thus $L(\ket{\psi})<3$ (respectively, $l(\ket{\psi})<3$)  \cite{walck2009only}, while $L(\ket{\ghz_3})=3$ (respectively, $l(\ket{\ghz_3})=3$).

{Since EDL and SDL are both integer-valued, they are nonconvex.  We also list two counterexamples as follows.
\begin{enumerate}[1.]
  \item  Let \begin{equation}
    \rho=\frac{1}{2}\ketbra{{\ghz}_3}{{\ghz}_3}+\frac{1}{2}\ketbra{D_3^1}{D_3^1}.
\end{equation} 
 We have shown that $l(\rho)=3$ in Sec.~\ref{sec:symmetric_EDL}. Since $l(\ket{\ghz_3})=3$ and $l(\ket{D_3^1})=2$, we have  $l(\rho)>\frac{1}{2}l(\ket{\ghz_3})+\frac{1}{2}l(\ket{D_3^1})$.
\item Let
\begin{equation}
\rho=\frac{1}{2}\ketbra{D_3^0}{D_3^0}+\frac{1}{2}\ketbra{D_3^2}{D_3^2}.
\end{equation}
According to Proposition~2, we have $L(\rho)=3$. Since $L(\ket{D_3^0})=1$ and $L(\ket{D_3^2})=2$, we have  $L(\rho)>\frac{1}{2}L(\ket{D_3^0})+\frac{1}{2}L(\ket{D_3^2})$.
\end{enumerate}}

Finally, we give a sufficient and necessary condition for $L(\rho)=1$.

\begin{lemma}\label{lemma:Lrho1}
For an $n$-qubit state $\rho$, $L(\rho)=1$ if and only if there exists $j$, $1\leq j\leq n$, such that $\rank{\rho_{\{j\}}}\geq 1$, and $\rank{\rho_{\{i\}}}= 1$ for all $1\leq i\neq j\leq n$. 
\end{lemma}
\begin{proof}
    Sufficiency. Firstly, we need to show a fact: for a bipartite state $\rho_{AB}$, if $\rank(\rho_A)=1$, then $\rho_{AB}=\rho_A\otimes \rho_B$. This is because $\cR(\rho_{AB})$ is a subspace of $\cR(\rho_A)\otimes\cR(\rho_B)$.
    Since $\rank{\rho_{\{j\}}}\geq 1$, and $\rank{\rho_{\{i\}}}= 1$ for all $1\leq i\neq j\leq n$, we know that $\rho=\ketbra{\psi_1}{\psi_1} \otimes \cdots\otimes \ketbra{\psi_{j-1}}{\psi_{j-1}}\otimes \rho_{\{j\}}\otimes \ketbra{\psi_{j+1}}{\psi_{j+1}}\otimes \cdots\otimes \ketbra{\psi_n}{\psi_n}$. Assume $\sigma\in \cC(\rho,\cS_1)$, then $\sigma_{\{i\}}=\rho_{\{i\}}$ for all $1\leq i\leq n$. Since $\sigma_{\{i\}}=\rho_{\{i\}}=\ketbra{\psi_i}{\psi_i}$ for all $1\leq i\neq j\leq n$, we have $\sigma=\ketbra{\psi_1}{\psi_1} \otimes \cdots\otimes \ketbra{\psi_{j-1}}{\psi_{j-1}}\otimes \sigma_{\{j\}}\otimes \ketbra{\psi_{j+1}}{\psi_{j+1}}\otimes \cdots\otimes \ketbra{\psi_n}{\psi_n}$ from the above fact. Note that  $\sigma_{\{j\}}=\rho_{\{j\}}$, then $\sigma=\rho$. Thus $\cC(\rho,\cS_1)$ contains only $\rho$, and $L(\rho)=1$.

    Necessity. Firstly, we need to prove a fact: for a bipartite state $\rho_{AB}$, if $\rank(\rho_A)\geq 2$, and $\rank(\rho_B)\geq 2$, then there exists a state $\sigma\neq \rho_A\otimes \rho_B$ such that $\sigma_A=\rho_A$ and $\sigma_B=\rho_B$. Assume the spectral decomposition of $\rho_A$ is $\rho_A=\sum_{i=1}^m\lambda_i\ketbra{\psi_i}{\psi_i}$, where $m\geq 2$ and $\lambda_i>0$ for all $1\leq i\leq m$. Since $\rank(\rho_B)\geq 2$, we can always find linearly independent states $\rho_1$ and $\rho_2$ such that $\lambda_1\rho_1+\lambda_2\rho_2=(\lambda_1+\lambda_2)\rho_B$. Let $\sigma=\lambda_1\ketbra{\psi_1}{\psi_1}\otimes\rho_1+\lambda_2\ketbra{\psi_2}{\psi_2}\otimes\rho_2+\sum_{i=3}^{m}\lambda_i\ketbra{\psi_i}{\psi_i}\otimes \rho_B$. We can verify that $\sigma_A=\rho_A$ and $\sigma_B=\rho_B$. Since $\ketbra{\psi_1}{\psi_1}$ and $\ketbra{\psi_2}{\psi_2}$ are linearly independent, and $\rho_1$ and $\rho_2$ are linearly independent, we have  $\sigma\neq \rho_A\otimes \rho_B$. Thus, the above fact is proved.  Now, we prove the necessity by contradiction. If there exist $i$, $j$,  $1\leq i\neq j\leq n$ such that 
    $\rank{\rho_{\{i\}}}\geq 2$ and $\rank{\rho_{\{j\}}}\geq 2$, then there exists $\sigma\in \cC(\rho, \cS_1)$, and $\sigma\neq \rho_{\{1\}}\otimes \rho_{\{2\}}\otimes \cdots\otimes \rho_{\{n\}}\in \cC(\rho, \cS_1)$ from the above fact. Therefore, $L(\rho)\geq 2$, which contradicts that $L(\rho)=1$.
\end{proof}
\vspace{0.4cm}
\section{II. \ The EDL and SDL of graph states}


Graph states are associated to simple graphs of the form  $G=([n],E)$, where $[n]$ is the set of vertices and $E $ is a set of edges.   Specifically,  the graph state $\ket{G}$ is the unique simultaneous eigenstate with eigenvalue +1 of the  matrices $\{M_i\}_{i=1}^n$ defined by 
 $M_i:=X_i\otimes\bigotimes_{j\in N_G(i)}Z_j$, where $X_i$ and $Z_j$ are Pauli matrices and $N_G(i)$ is the neighborhood of $i$, that is, the  set of all  vertices adjacent to $i$. We denote by  $\text{deg}_G(i)$ the degree of the vertex $i$,  that is, ${\rm deg}_G  (i)  : =  |  N_G(i)|$, and by $\Delta(G)$  the maximum degree of the graph $G$, that is, $\Delta(G)=\max_{i\in [n]}  \,  {\rm deg}_G  (i) $.
 
 It is known that a graph state is genuinely entangled   if and only if  the graph is connected  \cite{hein2004multiparty}.  In this case, we have the following bounds:

 \begin{lemma}
 For a connected graph  $G$  with $n\geq 3$ vertices,  one has the bounds
\begin{equation}\label{eq:Gone}
    3\leq L(\ket{G})   \leq 1  + \Delta(G)  . 
\end{equation}
The same bounds apply  to $l(|G\rangle)$.   
\end{lemma}
\begin{proof}
   If $G$ is a connected graph, then the compatibility set $\cC(\ket{G}, \cS_2)$ contains a fully separable state \cite{gittsovich2010multiparticle}, which implies $L(\ket{G})\geq 3$ (respectively, $l(\ket{G})\geq 3$). Since the compatibility set $\cC(\ket{G}, \{\{i\}\cup N_G(i)\}_{i=1}^{n}$) contains only the graph state $\ket{G}$  \cite{wu2015determination}, we have $L(\ket{G})\leq 1+ \Delta(G)$ (respectively, $l(\ket{G})\leq 1+ \Delta(G)$). 
\end{proof}

\vspace{0.4cm}

\begin{figure}[t]
		\centering		\includegraphics[scale=0.37]{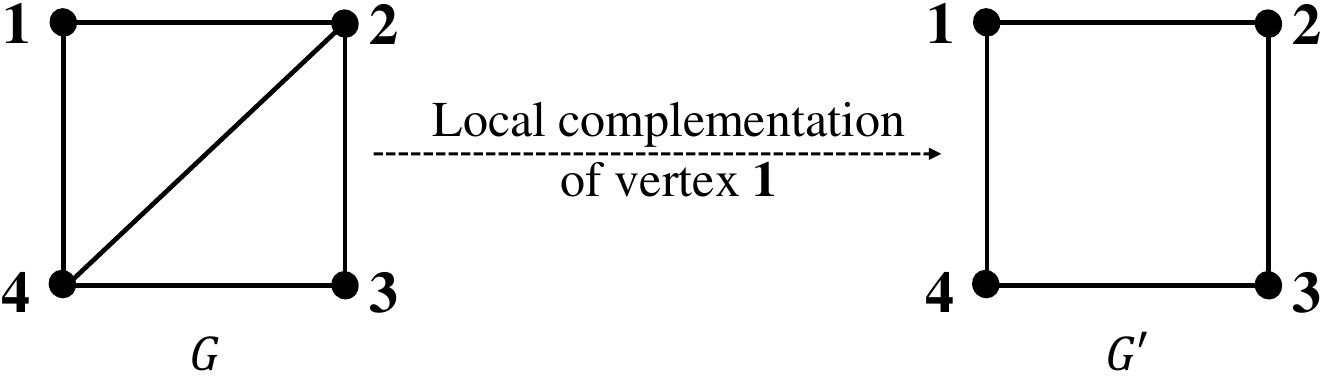}
		\caption{ {\bf Example of local complementation.}   The original graph $G$ undergoes local complementation with respect to the vertex 1,  which turns it into the cycle graph $G' =  C_4$.  
  } \label{fig:local}
\end{figure}

Eq.(\ref{eq:Gone}) determines the EDL  and SDL of all cluster states  and ring states 
corresponding to the cases where the graph $G$ is a line and a cycle, respectively.  In this case, one has $\Delta(G)=  2$, and therefore Eq.(\ref{eq:Gone})  implies $ L(|G \rangle)  =l(|G\rangle) =  3$.  

The upper bound in Eq.(\ref{eq:Gone}) can be further improved by recalling that the EDL is LU-invariant, and, in particular, is invariant under local Clifford (LC) unitaries. 

{LC-equivalence of graph states has a simple characterization in terms of an operation called local complementation \cite{hein2004multiparty,van2004graphical,adcock2020mapping}.    For a given vertex $i$,  local complementation of  $i$ consists in regarding the neighborhood  $N_G(i)$  as a subgraph of $G$, and replacing $N_G(i)$ with its complement graph, that is, the graph $N^{\rm c}_G(i)$  that has   the same vertices as $N_G(i)$, and has an edge between  two vertices $j$ and $k$  if and only if  $N_G(i)$ has no edge between them.   Refs. \cite{hein2004multiparty,van2004graphical,adcock2020mapping} proved that   two graph states $|G\rangle$ and $|G'\rangle$   are LC-equivalent if and only if   $G$ can be transformed into $G'$ through a sequence of local complementations.  In this case, we write $G'  \overset{\rm LC}{\sim} G$.    Since $L(|G'\rangle) =  L(|G\rangle)$ whenever $G'  \overset{\rm LC}{\sim} G$, Eq. (\ref{eq:Gone}) implies the bound  
\begin{equation}\label{eq:G}
    3\leq L(\ket{G}) \leq 1  +   \min_{G'  \overset{\rm LC}{\sim} G} \Delta(G'). 
\end{equation}
The same bounds apply  to $l(|G\rangle)$. 
For example, consider the graph $G$ in Fig.~\ref{fig:local}. By the local complementation of vextex $1$, we obtain a cycle  graph $G'$.  Then, Eq. (\ref{eq:G})  implies $ L(|G\rangle)=l(\ket{G})= 3$.}

 An $n$-qubit $k$-uniform state is a pure state such that all its  $k$-body marginals are maximally mixed \cite{scott2004multipartite}.   Since the $k$-body marginals are compatible with the maximally mixed $n$-qubit state,   the SDL (respectively, EDL) of (respectively, genuinely entangled) $k$-uniform states must be strictly larger than $k$. Graph states can be also used to construct $k$-uniform states \cite{helwig2013absolutely,sudevan2022n}.
For example, the $3$-uniform $6$-qubit genuinely entangled state  can be obtained from a $3$-regular graph \cite{helwig2013absolutely}, where every vertex has degree $3$, then we can determine that its SDL (respectively, EDL)  is $4$.


\section{III. Symmetric states}
A symmetric (bosonic) space $\cH_{Sym}$ is a subspace of $\cH_{[n]}$, where every pure state is invariant under the swap of any two subsystems, i.e.
\begin{equation}
  \cH_{Sym}:=\{\ket{\psi}\in \cH_{[n]}:\ {\swap}_{(i,j)}\ket{\psi}=\ket{\psi}, \forall \, 1\leq i,j\leq n \},  
\end{equation}
where $\swap_{(i,j)}$ is a swap operator that swaps   $i$-th and $j$-th subsystems.  Every pure state of  $\cH_{Sym}$ is called a  \emph{symmetric pure state}. 
Note that $\dim\cH_{Sym}=n+1$, and $\cH_{Sym}$ has an orthogonal basis $\{\ket{D_n^i}\}_{i=0}^n$ consisting of Dicke states, where 
\begin{equation}
 \ket{D_n^i}:=\frac{1}{\sqrt{\binom{n}{i}}}\sum_{\begin{subarray}{c}
   \\   s_j   \in  \{0,1\}, \, \forall \, j\in [n] \\  \\
     \sum_{j=1}^n  s_j=i 
     \end{subarray}}  |s_1\rangle \otimes |s_2\rangle \otimes \cdots\otimes | s_n\rangle  \, .
\end{equation}  
A symmetric pure state $\ket{\psi}$ can be written as
\begin{equation}
    \ket{\psi}=\sum_{i=0}^{n}a_{i}\ket{D_n^i},
\end{equation}
where $\sum_{i=0}^n|a_i|^2=1$.

Next, we introduce the concept of symmetric mixed states.  

\begin{definition}
   An $n$-qubit  state $\rho$ is a symmetric state if $\cR(\rho)\subseteq \cH_{Sym}$. In other words, $\rho=\swap_{(i,j)}\rho $ for all $1\leq i,j \leq n$.
\end{definition}

 A  symmetric  state can be written as $\rho=\sum_ip_i\ketbra{\psi_i}{\psi_i}$, where $\ket{\psi_i}\in \cH_{Sym}$. Furthermore, a symmetric  state $\rho$ can also be written as
\begin{equation}
    \rho=\sum_{i,j=0}^{n}a_{i,j}\ketbra{D_n^i}{D_n^j},
\end{equation}
where $A=(a_{i,j})_{0\leq i,j\leq n}$ is positive semidefinite and ${\sf Tr}(A)=1$. Specially, the \emph{diagonal symmetric state} is a mixture of the Dicke states,  i.e.
\begin{equation}
\rho=\sum_{i=0}^{n}\lambda_i\ket{D_n^i}\bra{D_n^i},
\end{equation}
where   $0\leq \lambda_i\leq 1$ for all $0\leq i\leq n$, and $\sum_{i=0}^n\lambda_i=1$.  

A symmetric state $\rho$ is either fully separable  or  genuinely entangled \cite{ichikawa2008exchange}. When $\rho$ is a fully separable  symmetric state, it can be written as $\rho=\sum_ip_i\ketbra{\alpha_i\alpha_i\cdots\alpha_i}{\alpha_i\alpha_i\cdots\alpha_i}$ \cite{ichikawa2008exchange}.

We stipulate that the  binomial coefficient $\binom{n}{i}=0$ if $i<0$ or $i>n$, and $\binom{0}{0}=1$. 
According to Ref.~ \cite{aloy2021quantum},
\begin{equation}\label{eq: Dicke_marginal}
\begin{aligned}
    \ketbra{D_n^i}{D_n^i}_{[k]}&=\sum_{s}\frac{\binom{k}{s}\binom{n-k}{i-s}}{\binom{n}{i}}\ketbra{D_k^s}{D_k^s},  \quad  \forall \,0\leq i\leq n; \\
\ketbra{D_n^i}{D_n^j}_{[k]}&=\sum_{s}\frac{\binom{n-k}{i-s}\sqrt{\binom{k}{s}\binom{k}{j-i+s}}}{\sqrt{\binom{n}{i}\binom{n}{j}}}
\ketbra{D_k^s}{D_k^{j-i+s}},  \quad \forall \, 0\leq i< j\leq n;\\
\ketbra{D_n^j}{D_n^i}_{[k]}&=\sum_{s}\frac{\binom{n-k}{i-s}\sqrt{\binom{k}{s}\binom{k}{j-i+s}}}{\sqrt{\binom{n}{i}\binom{n}{j}}}
\ketbra{D_k^{j-i+s}}{D_k^{s}},  \quad \forall \,0\leq i< j\leq n.
\end{aligned}
\end{equation}
Next, we give some properties for symmetric states.

\begin{lemma}\cite{i2016characterizing}\label{lem:rdm}
For an  $n$-qubit symmetric state $\rho$,   $\rho_S=\rho_{[k]}$ for every $S\in \cS_k$, and $\rho_{[k]}$ is also a symmetric state for every $2\leq k\leq n$.
\end{lemma}


\begin{lemma}\label{lem:rdm_coverse}
  For an $n$-qubit state $\rho$  and a collection $\cS$ of subsets of $[n]$, if $\rho_{S}$ is a symmetric state for every $S\in \cS$ and the hypergraph $G=([n], \cS)$ is connected, then $\rho$ is also a symmetric state.
\end{lemma}
\begin{proof}
  Firstly, we need to show that if $\rho_{\{i,j\}}$ is a symmetric state, then $\swap_{(i,j)}\rho=\rho$.
  
  We consider the $n$-qubit  pure state $\ket{\psi}$. By Schmidt decomposition,
 $\ket{\psi}$ can be written as 
 \begin{equation*}
\ket{\psi}=\sum_k\lambda_k\ket{\alpha_k}_{\{i,j\}}\ket{\beta_k}_{\overline{\{i,j\}}},
 \end{equation*}
where $\lambda_{k}\geq 0$, $\sum_{k}\lambda_{k}^2=1$, and the states   $\{\ket{\alpha_k}_{\{i,j\}}\}$ ($\{\ket{\beta_k}_{\overline{\{i,j\}}}\}$) are mutually orthogonal. 
 If  
\begin{equation}
  \ketbra{\psi}{\psi}_{\{i,j\}}=\sum_k\lambda_k^2\ket{\alpha_k}_{\{i,j\}}\bra{\alpha_k}
\end{equation}
is a symmetric state, then $\ket{\alpha_k}_{\{i,j\}}$ is a symmetric pure state for every $k$, i.e. $\swap_{(i,j)}\ket{\alpha_k}_{\{i,j\}}=\ket{\alpha_k}_{\{i,j\}}$ for every $k$. 
It implies that $\swap_{(i,j)}\ket{\psi}=\ket{\psi}$. 


Next, we consider the $n$-qubit mixed state  $\rho=\sum_kp_k\ketbra{\psi_k}{\psi_k}$. If $\rho_{\{i,j\}}=\sum_kp_k\ketbra{\psi_k}{\psi_k}_{\{i,j\}}$ is a symmetric state, then $\cR(\ketbra{\psi_k}{\psi_k}_{\{i,j\}})\subseteq\cR(\rho_{\{i,j\}})$. We obtain that $\ketbra{\psi_k}{\psi_k}_{\{i,j\}}$ is also a symmetric state for every $k$. From the above discussion, $\swap_{(i,j)}\ket{\psi_k}=\ket{\psi_k}$ for every $k$, which means that  $\swap_{(i,j)}\rho=\rho$. 

 If $\rho_{S}$ is a symmetric state, then $\rho_{\{s, t\}}$ is also a symmetric state for every  $\{s, t\}\subseteq S$ by Lemma~\ref{lem:rdm}. By the above discussion,
 $\swap_{(s,t)}\rho=\rho$ for every  $\{s, t\}\subseteq S$. If the hypergraph $G=([n], \cS)$ is connected, then for every $ \{a_1,a_m\}\subseteq [n]$, there exists a path 
\begin{equation}
    a_1\sim  a_2\sim  a_3\sim \cdots \sim a_m
\end{equation}
such that $\{a_j,a_{j+1}\}$ is contained in some $S$ for $1\leq j\leq m-1$, where $S\in \cS$. We obtain that $\swap_{(a_j,a_{j+1})}\rho=\rho$ for $1\leq j\leq m-1$. Since
\begin{equation}
    {\swap}_{(a_1,a_3)}={\swap}_{(a_1,a_2)}{\swap}_{(a_2,a_3)}{\swap}_{(a_1,a_2)},
\end{equation}
we have ${\swap}_{(a_1,a_3)}\rho=\rho$. Next,
since 
\begin{equation}
    {\swap}_{(a_1,a_4)}={\swap}_{(a_1,a_3)}{\swap}_{(a_3,a_4)}{\swap}_{(a_1,a_3)},
\end{equation}
we also have ${\swap}_{(a_1,a_4)}\rho=\rho$. By repeating this process
 $m-1$ times, we obtain ${\swap}_{(a_1,a_m)}\rho=\rho$. Thus ${\swap}_{(a_1,a_m)}\rho=\rho$ for every $ \{a_1,a_m\}\subseteq [n]$, i.e. $\rho$ is a symmetric state.
\end{proof}
\vspace{0.4cm}

\begin{lemma}\label{lemma: connected_symmetrc}
    For an $n$-qubit symmetric state $\rho$  and a collection $\cS$ of subsets of $[n]$, if  the hypergraph $G=([n], \cS)$ is connected, then every state $\sigma\in \cC(\rho, \cS)$ is a symmetric state. 
\end{lemma}
\begin{proof}
    For every $\sigma\in \cC(\rho, \cS)$, $\sigma_{S}=\rho_{S}$ is a symmetric state for every $S\in \cS$ by Lemma~\ref{lem:rdm}. Since $G=([n], \cS)$ is connected, $\sigma$ is also a symmetric state by Lemma~\ref{lem:rdm_coverse}.
\end{proof}
\vspace{0.4cm}

\section{IV. \ The proof of Proposition~1}\label{sec:symmetric_EDL}

\begin{proposition}
For an $n$-qubit entangled symmetric state $\rho$,   $l(\rho)=  \min  \{  k\mid   \rho_{[k]} \text{ is entangled} \}$.  
\end{proposition}

\begin{proof}
We stipulate that $\rho_{[1]}$ is separable, then we only need to prove the following the fact.

 Let $\rho$ be an $n$-qubit  entangled symmetric state,  then 
\begin{equation}\label{eq:case}
 l(\rho)   \begin{cases}
        \leq k, &  \text{if $\rho_{[k]}$ is  entangled};\\
        \geq k+1, & \text{if $\rho_{[k]}$ is  separable.}
    \end{cases}
\end{equation}

 Since $G=([n],\cS_k)$ is connected,    every $\sigma\in \cC(\rho, \cS_k)$ is a symmetric state according to Lemma~\ref{lemma: connected_symmetrc}. Note that $\sigma$ is either fully separable or genuinely entangled. If $\sigma_{[k]}=\rho_{[k]}$ is genuinely entangled, then $\sigma$ must be genuinely entangled. Thus $\cC(\rho, \cS_k)$ contains only genuinely entangled states, and $l(\rho)\leq k$.

  If $\rho_{[k]}$ is fully  separable, then for every $S\in \cS_k$, $\rho_{S}=\sum_{i}p_i\ket{\alpha_i}\ket{\alpha_i}\cdots \ket{\alpha_i}\bra{\alpha_i}\bra{\alpha_i}\cdots\bra{\alpha_i}$, where every $\ket{\alpha_i}\ket{\alpha_i}\cdots \ket{\alpha_i}$ is a fully separable pure state in $\cH_S$. We can find a fully separable state $\sigma=\sum_{i}p_i\ket{\alpha_i}\ket{\alpha_i}\cdots\ket{\alpha_i}\bra{\alpha_i}\bra{\alpha_i}\cdots\bra{\alpha_i}\in \cC(\rho, \cS_k)$,  where $\ket{\alpha_i}\ket{\alpha_i}\cdots \ket{\alpha_i}$ is a fully separable pure state in $\cH_{[n]}$. Thus  $l(\rho)\geq  k+1  $.

Thus, Proposition~\ref{prop:symmetric}  can be directly obtained from Eq.~\eqref{eq:case}.
\end{proof}
\vspace{0.4cm}


The most famous entanglement detection criterion is the so-called
positivity under partial transposition (PPT) criterion \cite{peres1996separability,horodecki2001separability}.
An $n$-qubit symmetric state is called   PPT (respectively, NPT)  if it is positive semidefinite (respectively, non-positive semidefinite) under
 the partial transpose of $\fl{n}{2}$ subsystems \cite{wolfe2014certifying,yu2016separability,tura2018separability}. The $2,3$-qubit symmetric states and the $n$-qubit diagonal symmetric  states are  fully separable if and only if they are PPT \cite{peres1996separability,peres1996separability,horodecki2001separability,eckert2002quantum,yu2016separability}. Based on PPT criterion, we can calculate the EDLs for a wide range of symmetric states. 
Specially, for a diagonal symmetric state, there exists a simple method to determine whether it is PPT or NPT \cite{quesada2017entanglement}.

\begin{lemma}\cite{quesada2017entanglement}\label{Lemma: hankel}
      Let $\rho=\sum_{i=0}^n\lambda_i\ketbra{D_n^i}{D_n^i}$, then $\rho$ is PPT if and only if the two Hankel matrices 
\begin{equation}\label{eq:hankel}
    M_0=\begin{pmatrix}
  p_0 & p_1 & p_2 &\cdots &p_{\fl{n}{2}}\\
  p_1 & p_2 & p_3 &\cdots &p_{\fl{n}{2}+1}\\
  p_2 & p_3 & p_4 &\cdots &p_{\fl{n}{2}+2}\\
   \vdots &  \vdots & \vdots & \ddots & \vdots\\
   p_{\fl{n}{2}} & p_{\fl{n}{2}+1} &   p_{\fl{n}{2}+2} &\ldots &    p_{2\fl{n}{2}}
\end{pmatrix},  \quad
 M_1=\begin{pmatrix}
  p_1 & p_2 & p_3 &\cdots &p_{\fl{n+1}{2}}\\
  p_2 & p_3 & p_4 &\cdots &p_{\fl{n+1}{2}+1}\\
  p_4 & p_5 & p_6 &\cdots &p_{\fl{n}{2}+2}\\
   \vdots &  \vdots & \vdots & \ddots & \vdots\\
   p_{\fl{n+1}{2}} & p_{\fl{n+1}{2}+1} &   p_{\fl{n+1}{2}+2} &\ldots &    p_{2\fl{n+1}{2}-1}
\end{pmatrix}  
\end{equation}
are positive semidefinite, where $p_i=\frac{1}{\binom{n}{i}}\lambda_i$ for $0\leq i\leq n$.  
\end{lemma}

Note that if $\rho$ is a diagonal symmetric state, then $\rho_{[k]}$ is also a diagonal symmetric state for $k\geq 2$. Next, we can provide a simple sufficient and necessary condition for $\rho_{[2]}$ and $\rho_{[3]}$ to be PPT.
\begin{lemma}
\label{lem:diganol_23}
Let $\rho=\sum_{i=0}^n\lambda_i\ketbra{D_n^i}{D_n^i}$, then $\rho_{[2]}$ is PPT if and only if
\begin{equation}\label{eq:l2}
     \left[\sum_{i=0}^{n}(n-i)(n-i-1)\lambda_i\right]\left[\sum_{i=0}^ni(i-1)\lambda_i\right]-\left[\sum_{i=0}^{n}i(n-i)\lambda_i\right]^2\geq 0,  
\end{equation}
and $\rho_{[3]}$  is PPT if and only if 
\begin{equation}\label{eq:l3}
\begin{aligned}
  \left[\sum_{i=0}^{n}(n-i)(n-i-1)(n-i-2)\lambda_i\right] \left[\sum_{i=0}^ni(i-1)(n-i)\lambda_i\right]- \left[\sum_{i=0}^{n}i(n-i)(n-i-1)\lambda_i\right]^2&\geq 0 ,\\ \text{and} \quad \left[\sum_{i=0}^{n}i(n-i)(n-i-1)\lambda_i\right]\left[\sum_{i=0}^ni(i-1)(i-2)\lambda_i\right]-\left[\sum_{i=0}^{n}i(i-1)(n-i)\lambda_i\right]^2&\geq  0.
  \end{aligned}
\end{equation}   
\end{lemma}
\begin{proof}
According to Eq.~\eqref{eq: Dicke_marginal}, we have
\begin{equation}
    \rho_{[k]}=\sum_{i=0}^n\lambda_i\ketbra{D_n^i}{D_n^i}_{[k]}=\sum_{i=0}^n\lambda_i\sum_{s}\frac{\binom{k}{s}\binom{n-k}{i-s}}{\binom{n}{i}}\ketbra{D_k^s}{D_k^s}=\sum_s\sum_{i=0}^n\lambda_i\frac{\binom{k}{s}\binom{n-k}{i-s}}{\binom{n}{i}}\ketbra{D_k^s}{D_k^s}.
\end{equation}
Next, we consider $k=2,3$.
\begin{enumerate}[(i)]
    \item When $k=2$, $\rho_{[2]}$ is PPT if and only if
\begin{equation}
    M_0=\begin{pmatrix}
      \sum_{i=0}^n\lambda_i\frac{\binom{n-2}{i}}{\binom{n}{i}} &  \sum_{i=0}^n\lambda_i\frac{\binom{n-2}{i-1}}{\binom{n}{i}} \\
     \sum_{i=0}^n\lambda_i\frac{\binom{n-2}{i-1}}{\binom{n}{i}}& \sum_{i=0}^n\lambda_i\frac{\binom{n-2}{i-2}}{\binom{n}{i}} \\
    \end{pmatrix}\geq 0, \quad M_1=\sum_{i=0}^n\lambda_i\frac{\binom{n-2}{i-1}}{\binom{n}{i}}\geq 0,
\end{equation}
i.e.,
\begin{equation}
   \left[\sum_{i=0}^{n}(n-i)(n-i-1)\lambda_i\right]\left[\sum_{i=0}^ni(i-1)\lambda_i\right]-\left[\sum_{i=0}^{n}i(n-i)\lambda_i\right]^2\geq 0.
\end{equation}
\item When $k=3$, $\rho_{[3]}$ is PPT if and only if
\begin{equation}
    M_0=\begin{pmatrix}
      \sum_{i=0}^n\lambda_i\frac{\binom{n-3}{i}}{\binom{n}{i}} &  \sum_{i=0}^n\lambda_i\frac{\binom{n-3}{i-1}}{\binom{n}{i}} \\
     \sum_{i=0}^n\lambda_i\frac{\binom{n-3}{i-1}}{\binom{n}{i}}& \sum_{i=0}^n\lambda_i\frac{\binom{n-3}{i-2}}{\binom{n}{i}} \\
    \end{pmatrix}\geq 0, \quad
      M_1=\begin{pmatrix}
     \sum_{i=0}^n\lambda_i\frac{\binom{n-3}{i-1}}{\binom{n}{i}}  &   \sum_{i=0}^n\lambda_i\frac{\binom{n-3}{i-2}}{\binom{n}{i}} \\
     \sum_{i=0}^n\lambda_i\frac{\binom{n-3}{i-2}}{\binom{n}{i}}& \sum_{i=0}^n\lambda_i\frac{\binom{n-3}{i-3}}{\binom{n}{i}} \\
    \end{pmatrix}\geq 0, 
\end{equation}
i.e.,
\begin{equation}
\begin{aligned}
  \left[\sum_{i=0}^{n}(n-i)(n-i-1)(n-i-2)\lambda_i\right] \left[\sum_{i=0}^ni(i-1)(n-i)\lambda_i\right]- \left[\sum_{i=0}^{n}i(n-i)(n-i-1)\lambda_i\right]^2&\geq 0,\\ 
  \text{and} \ \left[\sum_{i=0}^{n}i(n-i)(n-i-1)\lambda_i\right]\left[\sum_{i=0}^ni(i-1)(i-2)\lambda_i\right]-\left[\sum_{i=0}^{n}i(i-1)(n-i)\lambda_i\right]^2&\geq 0.
  \end{aligned}
\end{equation}
\end{enumerate}
This completes the proof.
\end{proof}
\vspace{0.4cm}


Next, we give several examples.
\begin{enumerate}[1.]
    \item Consider the Dicke state $\ket{D_n^k}$, where $n\geq 2$ and $1\leq k \leq n-1$. Since  $  \left[\sum_{i=0}^{n}(n-i)(n-i-1)\lambda_i\right]\left[\sum_{i=0}^ni(i-1)\lambda_i\right]-\left[\sum_{i=0}^{n}i(n-i)\lambda_i\right]^2=(n-k)(n-k-1)k(k-1)-k^2(n-k)^2<0$,
we know that $\ketbra{D_n^k}{D_n^k}_{[2]}$ is NPT. Thus we have $l(\ket{D_n^k})=2$.

\item Let
$\rho=\lambda_0\ket{D_n^0}\bra{D_n^0}+\lambda_1\ket{D_n^1}\bra{D_n^1}$, where $n\geq 2$ and $\lambda_1\neq 0$. Since $  \left[\sum_{i=0}^{n}(n-i)(n-i-1)\lambda_i\right]\left[\sum_{i=0}^ni(i-1)\lambda_i\right]-\left[\sum_{i=0}^{n}i(n-i)\lambda_i\right]^2=-(n-1)^2\lambda_1^2<0$, we know that $\rho_{[2]}$ is NPT. Thus we have $l(\rho)=2$.

\item Let
$\rho=\lambda_0\ket{D_n^0}\bra{D_n^0}+\lambda_2\ket{D_n^2}\bra{D_n^2}$, where $n\geq 3$ and $0<\lambda_2\leq \frac{n}{2n-2}$. Since $  \left[\sum_{i=0}^{n}(n-i)(n-i-1)\lambda_i\right]\left[\sum_{i=0}^ni(i-1)\lambda_i\right]-\left[\sum_{i=0}^{n}i(n-i)\lambda_i\right]^2=\left[n(n-1)\lambda_0+(n-2)(n-3)\lambda_2\right]\times 2\lambda_2-\left[2(n-2)\lambda_2)\right]^2=2(n-1)\lambda_2\left[n-(2n-2)\lambda_2\right]\geq 0$, we know that $\rho_{[2]}$ is PPT. Moreover, since $\left[\sum_{i=0}^{n}i(n-i)(n-i-1)\lambda_i\right]\left[\sum_{i=0}^ni(i-1)(i-2)\lambda_i\right]-\left[\sum_{i=0}^{n}i(i-1)(n-i)\lambda_i\right]^2=-4(n-2)^2\lambda_2^2< 0$, we know that $\rho_{[3]}$ is NPT. Thus we have $l(\rho)=3$.

\item Let $\rho=\frac{1}{24}\ketbra{D_4^0}{D_4^0}+\frac{1}{3}\ketbra{D_4^1}{D_4^1}+\frac{1}{12}\ketbra{D_4^2}{D_4^2}+\frac{1}{2}\ketbra{D_4^3}{D_4^3}+\frac{1}{24}\ketbra{D_4^4}{D_4^4}$. Since 
$  \left[\sum_{i=0}^{4}(4-i)(4-i-1)\lambda_i\right]$
$\times\left[\sum_{i=0}^4i(i-1)\lambda_i\right]-\left[\sum_{i=0}^{4}i(4-i)\lambda_i\right]^2=(12\times \frac{1}{24}+6\times \frac{1}{3}+2\times \frac{1}{12})(2\times \frac{1}{12}+6\times \frac{1}{2}+12\times \frac{1}{24})-(3\times\frac{1}{3}+4\times \frac{1}{12}+3\times \frac{1}{2})^2=\frac{7}{4}>0$,
and $\left[\sum_{i=0}^{4}i(4-i)(4-i-1)\lambda_i\right]\left[\sum_{i=0}^4i(i-1)(i-2)\lambda_i\right]-\left[\sum_{i=0}^{4}i(i-1)(4-i)\lambda_i\right]^2=(6\times \frac{1}{3}+4\times \frac{1}{12})(6\times \frac{1}{2}+24\times \frac{1}{24})-(4\times\frac{1}{12}+6\times \frac{1}{2})^2=-\frac{16}{9}<0$, we know that $\rho_{[2]}$ is PPT and $\rho_{[3]}$ is NPT. Thus
we have  $l(\rho)=3$.

\item Let  $\rho=\frac{1}{2}\ketbra{D_4^0}{D_4^0}+\frac{1}{3}\ketbra{D_4^1}{D_4^1}+\frac{1}{12}\ketbra{D_4^2}{D_4^2}+\frac{1}{24}\ketbra{D_4^3}{D_4^3}+\frac{1}{24}\ketbra{D_4^4}{D_4^4}$. Since 
$  \left[\sum_{i=0}^{4}(4-i)(4-i-1)(4-i-2)\lambda_i\right]$
$ \left[\sum_{i=0}^4i(i-1)(4-i)\lambda_i\right]- \left[\sum_{i=0}^{4}i(4-i)(4-i-1)\lambda_i\right]^2=(24\times \frac{1}{2}+6\times \frac{1}{3})(4\times \frac{1}{12}+6\times \frac{1}{24})-(6\times\frac{1}{3}+4\times \frac{1}{12})^2=\frac{49}{18}>0$,
and $\left[\sum_{i=0}^{4}i(4-i)(4-i-1)\lambda_i\right]\left[\sum_{i=0}^4i(i-1)(i-2)\lambda_i\right]-\left[\sum_{i=0}^{4}i(i-1)(4-i)\lambda_i\right]^2=(6\times \frac{1}{3}+4\times \frac{1}{12})(6\times \frac{1}{24}+24\times \frac{1}{24})-(4\times\frac{1}{12}+6\times \frac{1}{24})^2=\frac{371}{144}>0$, we know that $\rho_{[3]}$ is PPT. However,
\begin{equation}
    M_0=\begin{pmatrix}
        \frac{1}{2} & \frac{1}{12} & \frac{1}{72}\\
         \frac{1}{12} & \frac{1}{72} &\frac{1}{96}\\
 \frac{1}{72} &\frac{1}{96} & \frac{1}{24}    
    \end{pmatrix}<0,
\end{equation}
which implies that $\rho$ is NPT. Thus we have  $l(\rho)=4$.

\item  Let   $\rho=\frac{1}{2}\ketbra{{\ghz}_3}{{\ghz}_3}+\frac{1}{2}\ketbra{D_3^1}{D_3^1}$, where $\ket{{\ghz}_3}=\frac{1}{\sqrt{2}}(\ket{D_3^0}+\ket{D_3^3})$. Consider $\rho_{[2]}=\frac{5}{12}\ketbra{D_2^0}{D_2^0}+\frac{1}{3}\ketbra{D_2^1}{D_2^1}+\frac{1}{4}\ketbra{D_2^2}{D_2^2}$. Since $\frac{5}{12}\times \frac{1}{4}-\frac{1}{36}=\frac{11}{144}>0$, $\rho_{[2]}$ is PPT. Moreover, since $\rho$ is NPT, we have $l(\rho)=3$. 
\end{enumerate}


\section{V. \ The proof of Proposition 2}

Firstly, we show that
$L(\ket{D_n^k})=2$. Note that $L(\ket{D_n^k})=2$ can be also obtained from Ref.~\cite{chen2012correlations}, which shows that $\ket{D_n^k}$ is the unique ground state of a 2-local Hamiltonian. Here we give a different proof. This method will be applied to calculate the SDLs of more general diagonal symmetric states.

\begin{lemma}\cite{chen2012correlations}
  For all $1\leq k\leq n-1$ and $n\geq 2$, we have 
    $L(\ket{D_n^k})=2$. 
 \end{lemma}
 \begin{proof}
 Since $G=([n],\cS_2)$ is connected,  every $\sigma\in \cC(\ket{D_n^k}, \cS_2)$ is a symmetric state according to Lemma~\ref{lemma: connected_symmetrc}.
 We can   
write  $\sigma$  as
\begin{equation}
    \sigma=\sum_{i,j=0}^{n}a_{i,j}\ketbra{D_n^i}{D_n^j},
\end{equation}
where $A=(a_{i,j})_{0\leq i,j\leq n}$ is positive semidefinite and ${\sf Tr}(A)=1$. 
According to Eq.~\eqref{eq: Dicke_marginal}, we have
\begin{equation}\label{eq: Dicke_marginal_2}
\begin{aligned}
    \ketbra{D_n^i}{D_n^i}_{[2]}&=\sum_{s}\frac{\binom{2}{s}\binom{n-2}{i-s}}{\binom{n}{i}}\ketbra{D_2^s}{D_2^s},  \quad  \forall \,0\leq i\leq n; \\
\ketbra{D_n^i}{D_n^j}_{[2]}&=\sum_{s}\frac{\binom{n-2}{i-s}\sqrt{\binom{2}{s}\binom{2}{j-i+s}}}{\sqrt{\binom{n}{i}\binom{n}{j}}}
\ketbra{D_2^s}{D_2^{j-i+s}},  \quad \forall \, 0\leq i< j\leq n;\\
\ketbra{D_n^j}{D_n^i}_{[2]}&=\sum_{s}\frac{\binom{n-2}{i-s}\sqrt{\binom{2}{s}\binom{2}{j-i+s}}}{\sqrt{\binom{n}{i}\binom{n}{j}}}
\ketbra{D_2^{j-i+s}}{D_2^{s}},  \quad \forall \,0\leq i< j\leq n.
\end{aligned}
\end{equation}
Since $\sigma\in \cC(\ket{D_n^k}, \cS_2)$, we have $\sigma_{[2]}= \ketbra{D_n^k}{D_n^k}_{[2]}$, i.e.,
\begin{equation}\label{eq:sum_dicke}
   \sum_{i,j=0}^{n}a_{i,j}\ketbra{D_n^i}{D_n^j}_{[2]}=\sum_{i=0}^na_{i,i}\ketbra{D_n^i}{D_n^i}_{[2]}+\sum_{0\leq i< j\leq n}a_{i,j}\ketbra{D_n^i}{D_n^j}_{[2]}+\sum_{0\leq i< j\leq n}a_{j,i}\ketbra{D_n^j}{D_n^i}_{[2]}=\ketbra{D_n^k}{D_n^k}_{[2]}.
\end{equation}
Substituting Eq.~\eqref{eq: Dicke_marginal_2} into Eq.~\eqref{eq:sum_dicke}, we have

\begin{equation}\label{eq:non_off}
\begin{aligned}
 &\sum_{i=0}^na_{i,i}\sum_{s}\frac{\binom{2}{s}\binom{n-2}{i-s}}{\binom{n}{i}}\ketbra{D_2^s}{D_2^s}+\sum_{0\leq i< j\leq n}a_{i,j}\sum_{s}\frac{\binom{n-2}{i-s}\sqrt{\binom{2}{s}\binom{2}{j-i+s}}}{\sqrt{\binom{n}{i}\binom{n}{j}}}
\ketbra{D_2^{s}}{D_2^{j-i+s}}\\
&+\sum_{0\leq i< j\leq n}a_{j,i}\sum_{s}\frac{\binom{n-2}{i-s}\sqrt{\binom{2}{s}\binom{2}{j-i+s}}}{\sqrt{\binom{n}{i}\binom{n}{j}}}
\ketbra{D_2^{j-i+s}}{D_2^{s}}=\sum_{s}\frac{\binom{2}{s}\binom{n-2}{k-s}}{\binom{n}{k}}\ketbra{D_2^s}{D_2^s}.
 \end{aligned}   
\end{equation}
{Since there are no off-diagonal elements on the right side of Eq.~\eqref{eq:non_off} in the Dicke basis $\{\ket{D_2^s}\}_{s=0}^2$}, we have
\begin{equation}
    \sum_{0\leq i< j\leq n}a_{i,j}\sum_{s}\frac{\binom{n-2}{i-s}\sqrt{\binom{2}{s}\binom{2}{j-i+s}}}{\sqrt{\binom{n}{i}\binom{n}{j}}}
\ketbra{D_2^{s}}{D_2^{j-i+s}}=
\sum_{0\leq i< j\leq n}a_{j,i}\sum_{s}\frac{\binom{n-2}{i-s}\sqrt{\binom{2}{s}\binom{2}{j-i+s}}}{\sqrt{\binom{n}{i}\binom{n}{j}}}
\ketbra{D_2^{j-i+s}}{D_2^{s}}=\textbf{0},
\end{equation}
and
\begin{equation}\label{eq:dig_equal_2}
\sum_{i=0}^na_{i,i}\sum_{s}\frac{\binom{2}{s}\binom{n-2}{i-s}}{\binom{n}{i}}\ketbra{D_2^s}{D_2^s}=\sum_{s}\frac{\binom{2}{s}\binom{n-2}{k-s}}{\binom{n}{k}}\ketbra{D_2^s}{D_2^s}.
\end{equation}
Eq.~\eqref{eq:dig_equal_2} is equivalent to
\begin{equation}\label{eq:dig_equal_23}
\sum_{s}\sum_{i=0}^na_{i,i}\frac{\binom{2}{s}\binom{n-2}{i-s}}{\binom{n}{i}}\ketbra{D_2^s}{D_2^s}=\sum_{s}\frac{\binom{2}{s}\binom{n-2}{k-s}}{\binom{n}{k}}\ketbra{D_2^s}{D_2^s}.
\end{equation}
Then we have 
\begin{equation}
\sum_{i=0}^na_{i,i}\frac{\binom{2}{s}\binom{n-2}{i-s}}{\binom{n}{i}}=\frac{\binom{2}{s}\binom{n-2}{k-s}}{\binom{n}{k}}, \quad \forall \,0\leq s\leq 2.
\end{equation}
When $s=0,1,2$ respectively,   we have
\begin{equation}
   \sum_{i=0}^na_{i,i}\frac{\binom{n-2}{i}}{\binom{n}{i}}=\frac{\binom{n-2}{k}}{\binom{n}{k}}, \quad   \sum_{i=0}^na_{i,i}\frac{\binom{n-2}{i-1}}{\binom{n}{i}}=\frac{\binom{n-2}{k-1}}{\binom{n}{k}}, \quad
  \sum_{i=0}^na_{i,i}\frac{\binom{n-2}{i-2}}{\binom{n}{i}}=\frac{\binom{n-2}{k-2}}{\binom{n}{k}} .
\end{equation}
These three equations can be simplified into the following three equations:
\begin{align}
   &\sum_{i=0}^{n}a_{i,i}(n-i)(n-i-1)=(n-k)(n-k-1), \label{eq:1} \\
   &\sum_{i=0}^{n}a_{i,i}i(n-i)=k(n-k), \label{eq:2}\\
  & \sum_{i=0}^{n}a_{i,i}i(i-1)=k(k-1) . \label{eq:3}
\end{align}

Next,  by adding both sides of Eq.~\eqref{eq:1} and Eq.~\eqref{eq:2}, and dividing both sides by $n-1$  (i.e. (Eq.~\eqref{eq:1} + Eq.~\eqref{eq:2})/$(n-1)$),  we obtain the following equation,
\begin{equation}\label{eq:n-k}
    \sum_{i=0}^{n}a_{i,i}(n-i)=(n-k).
\end{equation}
 The equation (Eq.~\eqref{eq:2} + Eq.~\eqref{eq:3})/$(n-1)$ is
\begin{equation}\label{eq:k}
    \sum_{i=0}^{n}a_{i,i}i=k.
\end{equation}
 The equation Eq.~\eqref{eq:1} + Eq.~\eqref{eq:n-k} is
\begin{equation}\label{eq:(n-k)^2}
    \sum_{i=0}^{n}a_{i,i}(n-i)^2=(n-k)^2.
\end{equation}
 The equation Eq.~\eqref{eq:3} + Eq.~\eqref{eq:k} is
\begin{equation}\label{eq:k^2}
    \sum_{i=0}^{n}a_{i,i}i^2=k^2.
\end{equation}
Then 
\begin{equation}
  k^2(n-k)^2=\left[\sum_{i=0}^{n}a_{i,i}i^2\right]\left[\sum_{i=0}^{n}a_{i,i}(n-i)^2\right]\geq 
  \left[\sum_{i=0}^{n}a_{i,i}i(n-i)\right]^2=k^2(n-k)^2
\end{equation}
where the inequality is the Cauchy–Schwarz inequality, and the last equation is obtained from Eq.~\eqref{eq:2}.   Since the Cauchy–Schwarz inequality has the``=" sign, there must exist a real number $s\neq 0$, such that
\begin{equation}
  \sqrt{a_{i,i}}i= s\sqrt{a_{i,i}}(n-i), \quad  \forall \,0\leq i\leq n. 
\end{equation}
Then $a_{0,0}=a_{n,n}=0$, and there exists only one nonzero $a_{j,j}$ for all $1\leq j\leq n-1$ (i.e. $a_{j,j}\neq 0$, and $a_{i,i}=0$ for all $1\leq i\neq j\leq n-1$). Moreover, since ${\sf Tr}(A)=1$, we have $\sum_{i=0}^na_{i,i}=1$. It implies that $a_{j,j}=1$, and $a_{i,i}=0$ for all $0\leq i\neq j\leq n$. By using Eq.~\eqref{eq:k}, we obtain that $j=k$, that is $a_{k,k}=1$, and $a_{i,i}=0$ for all $0\leq i\neq k\leq n$. Since $A$ is positive semidefinite, it must have $a_{i,j}=0$ for all $0\leq i\neq j\leq n$ ({this is obtained from the property of a positive semidefinite matrix  $A=(a_{i,j})$:
if $a_{i,i}=0$, then $a_{i,j}=a_{j,i}=0$ for every $j$}). Thus $\sigma=\ketbra{D_n^k}{D_n^k}$.

Above all,  the  compatibility set $\cC(\ket{D_n^k}, \cS_2)$ contains only $\ket{D_n^k}$, which implies
 $L(\ket{D_n^k})=2$.
\end{proof}
\vspace{0.4cm}

In order to prove Proposition~\ref{pro:symmetric_SDL}, we need to give several Lemmas.

\begin{lemma}\label{lemma:nonhom_general}
Let $\rho=\sum_{i=0}^{n}\lambda_i\ket{D_n^i}\bra{D_n^i}$, and $1\leq k\leq n$. If the following nonhomogeneous linear equations  respect to $(a_{0,0},a_{1,1},\ldots,a_{n,n})$ $( a_{i,i}\geq 0$ for $0\leq i\leq n)$ 
\begin{equation}\label{eq:nonhomo_general}
  \sum_{i=0}^na_{i,i}\frac{\binom{k}{s}\binom{n-k}{i-s}}{\binom{n}{i}}=\sum_{i=0}^n\lambda_i\frac{\binom{k}{s}\binom{n-k}{i-s}}{\binom{n}{i}},  \quad \forall \, 0\leq s\leq k,  
\end{equation}
have a solution $(a_{0,0},a_{1,1},\ldots,a_{n,n})\neq (\lambda_0,\lambda_1,\ldots,\lambda_n)$, then $L(\rho)\geq k+1$.
\end{lemma}
\begin{proof}
Note that 
$(\lambda_0,\lambda_1,\ldots,\lambda_n)$ is always a solution of the nonhomogeneous linear equations in Eq.~\eqref{eq:nonhomo_general}. 
Assume there exists a solution $(a_{0,0},a_{1,1},\ldots,a_{n,n})\neq (\lambda_0,\lambda_1,\ldots,\lambda_n)$ for the nonhomogeneous linear equations in Eq.~\eqref{eq:nonhomo_general}.
Since
\begin{equation}
    \sum_{s=0}^k\frac{\binom{k}{s}\binom{n-k}{i-s}}{\binom{n}{i}}=   \sum_{s=0}^i\frac{\binom{k}{s}\binom{n-k}{i-s}}{\binom{n}{i}}=\frac{\binom{n}{i}}{\binom{n}{i}}=1.
\end{equation}
we have
\begin{equation}
  \sum_{i=0}^na_{i,i}=\sum_{i=0}^na_{i,i}\sum_{s=0}^k\frac{\binom{k}{s}\binom{n-k}{i-s}}{\binom{n}{i}}=\sum_{s=0}^k\sum_{i=0}^na_{i,i}\frac{\binom{k}{s}\binom{n-k}{i-s}}{\binom{n}{i}}=\sum_{s=0}^k\sum_{i=0}^n\lambda_i\frac{\binom{k}{s}\binom{n-k}{i-s}}{\binom{n}{i}}=\sum_{i=0}^n\lambda_i\sum_{s=0}^k\frac{\binom{k}{s}\binom{n-k}{i-s}}{\binom{n}{i}}=\sum_{i=0}^n\lambda_i=1. 
\end{equation}
Let $\sigma=\sum_{i=0}^{n}a_{i,i}\ket{D_n^i}\bra{D_n^i}$, then $\sigma\neq \rho$. We claim that $\sigma\in \cC(\rho,\cS_k)$. 

By Lemma~\ref{lem:rdm}, we only need to show  $\sigma_{[k]}=\rho_{[k]}$. According to Eq.~\eqref{eq: Dicke_marginal}, we have
\begin{equation}
      \ketbra{D_n^i}{D_n^i}_{[k]}=\sum_{s}\frac{\binom{k}{s}\binom{n-k}{i-s}}{\binom{n}{i}}\ketbra{D_k^s}{D_k^s}=\sum_{s=0}^k\frac{\binom{k}{s}\binom{n-k}{i-s}}{\binom{n}{i}}\ketbra{D_k^s}{D_k^s},  \quad  \forall \,0\leq i\leq n.
\end{equation}
Then  $\sigma_{[k]}=\rho_{[k]}$ is equivalent to 
\begin{equation}
\sum_{i=0}^na_{i,i}\sum_{s=0}^k\frac{\binom{k}{s}\binom{n-k}{i-s}}{\binom{n}{i}}\ketbra{D_k^s}{D_k^s}=\sum_{i=0}^n\lambda_i\sum_{s=0}^k\frac{\binom{k}{s}\binom{n-k}{i-s}}{\binom{n}{i}}\ketbra{D_k^s}{D_k^s}, 
\end{equation}
i.e.,
\begin{equation}\label{eq:hold}
\sum_{s=0}^k\sum_{i=0}^na_{i,i}\frac{\binom{k}{s}\binom{n-k}{i-s}}{\binom{n}{i}}\ketbra{D_k^s}{D_k^s}=      \sum_{s=0}^k\sum_{i=0}^n\lambda_i\frac{\binom{k}{s}\binom{n-k}{i-s}}{\binom{n}{i}}\ketbra{D_k^s}{D_k^s}.
\end{equation}
By Eq.~\eqref{eq:nonhomo_general}, we know that Eq.~\eqref{eq:hold}  always holds.
Thus $\rho\neq \sigma\in\cC(\rho,\cS_k)$, which means $L(\rho)\geq k+1$. 
\end{proof}
\vspace{0.4cm}

For a state $\rho=\sum_{i=0}^{n}\lambda_i\ket{D_n^i}\bra{D_n^i}$,
 if the nonhomogeneous linear equations of Eq.~\eqref{eq:nonhomo_general} have only one solution $(a_{0,0},a_{1,1},\ldots,a_{n,n})= (\lambda_0,\lambda_1,\ldots,\lambda_n)$, then   $L(\rho)\leq k$ is not always true. For example, assume $\lambda_i=0$ for all $1\leq i\leq n-1$ and $\lambda_0,\lambda_n\neq 0$, then $\rho=\lambda_0\ket{D_n^0}\bra{D_n^0}+\lambda_n\ket{D_n^n}\bra{D_n^n}$. When $n\geq 3$, and $k=n-1$, the nonhomogeneous linear equations of Eq.~\eqref{eq:nonhomo_general} have only one solution $(a_{0,0},a_{1,1},\ldots,,a_{n-1,n-1},a_{n,n})= (\lambda_0,0,\ldots,0,\lambda_n)$.  However, the pure state $\sqrt{\lambda_0}\ket{D_n^0}+\sqrt{\lambda_1}\ket{D_n^n} \in \cC(\rho, \cS_{n-1})$, which means that $L(\rho)=n$. However, if $\lambda_i=0$ for $k+1\leq i\leq n$, we would obtain  $L(\rho)\leq k$. See Lemma~\ref{lem:nonhom}.

\begin{lemma}\label{lemma:nonhom_special}
The general solution of the nonhomogeneous linear equations in Eq.~\eqref{eq:nonhomo_general} is
\begin{equation}\label{eq:nonhom_special}
\begin{aligned}
a_{r,r}&=\sum_{i=k+1}^n(-1)^{k-r+1}\binom{i}{k}\binom{k}{r}\frac{i-k}{i-r}s_i+\lambda_r,  \quad \forall \,0\leq r\leq k,\\
a_{r,r}&=s_r+\lambda_r, \quad \forall \,k+1\leq r\leq n,
\end{aligned}
\end{equation}
where $s_r\in \mathbb{R}$ for all $k+1\leq r\leq n$. 
\end{lemma}
\begin{proof}
 Note that  $(\lambda_0,\lambda_1,\ldots,\lambda_n)$ is a solution of the nonhomogeneous linear equations in Eq.~\eqref{eq:nonhomo_general}.  Then we only need to solve the following homogeneous linear equations  respect to $(x_0,x_1,\ldots,x_n)$,
\begin{equation}\label{eq:eqlem}
  \sum_{i=0}^nx_{i}\frac{\binom{k}{s}\binom{n-k}{i-s}}{\binom{n}{i}}=0,  \quad \forall \,0\leq s\leq k.  
\end{equation}
By direct calculation, we obtain the general solution of Eq.~\eqref{eq:eqlem},
\begin{equation}
\begin{aligned}
x_{r}&=\sum_{i=k+1}^n(-1)^{k-r+1}\binom{i}{k}\binom{k}{r}\frac{i-k}{i-r}s_i,  \quad \forall \,0\leq r\leq k,\\
x_{r}&=s_r, \quad \forall \,k+1\leq r\leq n,
\end{aligned}
\end{equation}
where $s_r\in \mathbb{R}$ for all $k+1\leq r\leq n$. Therefore, we  
obtain  the general solution of the nonhomogeneous linear equations in Eq.~\eqref{eq:nonhomo_general}, i.e., 
\begin{equation}
\begin{aligned}
a_{r,r}&=x_r+\lambda_r=\sum_{i=k+1}^n(-1)^{k-r+1}\binom{i}{k}\binom{k}{r}\frac{i-k}{i-r}s_i+\lambda_r,  \quad \forall \,0\leq r\leq k,\\
a_{r,r}&=x_r+\lambda_r=s_r+\lambda_r, \quad \forall \,k+1\leq r\leq n,
\end{aligned}
\end{equation}
where $s_r\in \mathbb{R}$ for all $k+1\leq r\leq n$. This completes the proof. 
\end{proof}
\vspace{0.4cm}

\begin{lemma}\label{lemma:nondiag}
If $\sum_{0\leq i<j\leq k}a_{i,j}\ketbra{D_n^i}{D_n^j}_{[k]}=\textbf{0}$, where $1\leq k\leq n$, then $a_{i,j}=0$ for all $0\leq i<j\leq k$.
\end{lemma}
\begin{proof}
According to Eq.~\eqref{eq: Dicke_marginal}, we have \begin{equation}
\ketbra{D_n^i}{D_n^j}_{[k]}=\sum_{s}\frac{\binom{n-k}{i-s}\sqrt{\binom{k}{s}\binom{k}{j-i+s}}}{\sqrt{\binom{n}{i}\binom{n}{j}}}
\ketbra{D_k^s}{D_k^{j-i+s}},  \quad \forall \,0\leq i< j\leq k.
\end{equation}
where $s$ in the sum is from $\max\{0, k+i-n\}$ to $i$, and the sum must contain 
$\frac{\sqrt{\binom{k}{i}\binom{k}{j}}}{\sqrt{\binom{n}{i}\binom{n}{j}}}
\ketbra{D_k^i}{D_k^{j}}$ $(s=i)$.
Let $\rho=\sum_{0\leq i<j\leq k}a_{i,j}\ketbra{D_n^i}{D_n^j}_{[k]}=\textbf{0}$, then
\begin{equation}
\rho=\sum_{0\leq i<j\leq k}a_{i,j}\sum_{s}\frac{\binom{n-k}{i-s}\sqrt{\binom{k}{s}\binom{k}{j-i+s}}}{\sqrt{\binom{n}{i}\binom{n}{j}}}
\ketbra{D_k^s}{D_k^{j-i+s}}=\textbf{0}.
\end{equation}

 \noindent \textbf{Step 1:} 
  \begin{equation}
    a_{k-1,k}=\bra{D_k^{k-1}}\rho\ket{D_k^{k}}=0.
\end{equation}

\noindent  \textbf{Step 2:}
\begin{equation}
    a_{k-2,k-1}=\bra{D_k^{k-2}}\rho\ket{D_k^{k-1}}=0,  \quad    a_{k-2,k}=\bra{D_k^{k-2}}\rho\ket{D_k^{k}}=0.
\end{equation}

\noindent  \textbf{Step 3:}
\begin{equation}
    a_{k-3,k-2}=\bra{D_k^{k-3}}\rho\ket{D_k^{k-2}}=0,  \quad       a_{k-3,k-1}=\bra{D_k^{k-3}}\rho\ket{D_k^{k-1}}=0,  \quad       a_{k-3,k}=\bra{D_k^{k-3}}\rho\ket{D_k^{k}}=0.
\end{equation}

\noindent  \textbf{Step $j$ $(4\leq j\leq k$): }
\begin{equation}
    a_{k-j,r}=\bra{D_k^{k-j}}\rho\ket{D_k^{r}}=0,  \quad  \forall \,k-j+1\leq r\leq k.
\end{equation}
Above all,  $a_{i,j}=0$ for all  $0\leq i<j\leq k$.
\end{proof}

\begin{lemma}\label{lem:nonhom}
Let $\rho=\sum_{i=0}^k\lambda_i\ket{D_n^i}\bra{D_n^i}$, where $1\leq k\leq n$, then $L(\rho)\leq k$ if and only if
the following nonhomogeneous linear equations  respect to $(a_{0,0},a_{1,1},\ldots,a_{n,n})$ $( a_{i,i}\geq 0$ for $0\leq i\leq n)$ 
\begin{equation}\label{eq:noneqlem}
  \sum_{i=0}^na_{i,i}\frac{\binom{k}{s}\binom{n-k}{i-s}}{\binom{n}{i}}=\sum_{i=0}^k\lambda_i\frac{\binom{k}{s}\binom{n-k}{i-s}}{\binom{n}{i}},  \quad \forall \,0\leq s\leq k,  
\end{equation}
have only one solution $(a_{0,0},a_{1,1},\ldots,a_{k,k},a_{k+1,k+1}, \ldots, a_{n,n})=(\lambda_0,\lambda_1,\ldots,\lambda_k,0,\ldots,0)$.
\end{lemma}
\begin{proof}
The necessity is obtained from Lemma~\ref{lemma:nonhom_general}. We only need to prove the sufficiency.

Assume $k=1$. According to Lemma~\ref{lemma:nonhom_special}, the general solution of the nonhomogeneous linear
equation of Eq.~\eqref{eq:noneqlem}  have the following the solution,
\begin{equation}
    \begin{aligned}
        a_{0,0}&=\sum_{i=2}^n(i-1)s_i+\lambda_0,\\
        a_{1,1}&=-\sum_{i=2}^nis_i+\lambda_1,\\
       a_{r,r}&=s_r, \ \forall \, 2\leq r\leq n,\\
    \end{aligned}
\end{equation}
where $s_r\in \mathbb{R}$ for all $2\leq r\leq n$, and $a_{r,r}\geq 0$ for all $0\leq r\leq n$.
If $\lambda_1> 0$, then as long as $s_r>0$ is small enough for all $2\leq r\leq n$, we can always find a solution $(a_{0,0},a_{1,1},\ldots,a_{k,k},a_{k+1,k+1}, \ldots, a_{n,n})\neq (\lambda_0,\lambda_1, 0\ldots,0,0,\ldots,0)$  with $a_{r,r}\geq 0$ for all $0\leq r\leq n$.
This means that $\lambda_1=0$. In this case, $a_{1,1}=-\sum_{i=2}^nis_i=-\sum_{i=2}^nia_{i,i}\geq 0$, which implies that $ a_{r,r}=s_r=0$ for all $2\leq r\leq n$, and    $(a_{0,0},a_{1,1},\ldots,a_{k,k},a_{k+1,k+1}, \ldots, a_{n,n})=(\lambda_0,0,\ldots,0,0,\ldots,0)$. Thus  the nonhomogeneous linear equations of Eq.~\eqref{eq:noneqlem} have only one solution is equivalent to that $\lambda_1=0$. According Lemma~\ref{lemma:Lrho1}, $L(\ket{{D_n^0}})=1$.

Assume $k\geq 2$. Since $G=([n],\cS_k)$ is connected, every $\sigma\in \cC(\rho,\cS_k)$ is a symmetric state according to Lemma~\ref{lemma: connected_symmetrc}.
We can
write  $\sigma$  as
\begin{equation}
\sigma=\sum_{i,j=0}^{n}a_{i,j}\ketbra{D_n^i}{D_n^j},
\end{equation}
where $A=(a_{i,j})_{0\leq i,j\leq n}$ is positive semidefinite and ${\sf Tr}(A)=1$. 
Since  $\sigma\in \cC(\rho,\cS_k)$, we have $\sigma_{[k]}=\rho_{[k]}$, i.e., 

\begin{equation}\label{eq:equal_marginal}
\sum_{i,j=0}^{n}a_{i,j}\ketbra{D_n^i}{D_n^j}_{[k]}=\sum_{i=0}^k\lambda_i\ket{D_n^i}\bra{D_n^i}_{[k]}.
\end{equation}
Substituting Eq.~\eqref{eq: Dicke_marginal} into Eq.~\eqref{eq:equal_marginal}, we have
\begin{equation}\label{eq:non_off1}
\begin{aligned}
 &\sum_{i=0}^na_{i,i}\sum_{s}\frac{\binom{k}{s}\binom{n-k}{i-s}}{\binom{n}{i}}\ketbra{D_k^s}{D_k^s}+\sum_{0\leq i< j\leq n}a_{i,j}\sum_{s}\frac{\binom{n-k}{i-s}\sqrt{\binom{k}{s}\binom{k}{j-i+s}}}{\sqrt{\binom{n}{i}\binom{n}{j}}}
\ketbra{D_k^{s}}{D_k^{j-i+s}}\\
&+\sum_{0\leq i< j\leq n}a_{j,i}\sum_{s}\frac{\binom{n-k}{i-s}\sqrt{\binom{k}{s}\binom{k}{j-i+s}}}{\sqrt{\binom{n}{i}\binom{n}{j}}}
\ketbra{D_k^{j-i+s}}{D_k^{s}}=\sum_{i=0}^k\lambda_i\sum_{s}\frac{\binom{k}{s}\binom{n-k}{i-s}}{\binom{n}{i}}\ketbra{D_k^s}{D_k^s}.
 \end{aligned}   
\end{equation}
Since there are no off-diagonal elements on the right side of Eq.~\eqref{eq:non_off1} in the Dicke basis $\{\ket{D_k^s}\}_{s=0}^k$, we have
\begin{equation}\label{eq:nondiag}
    \sum_{0\leq i< j\leq n}a_{i,j}\sum_{s}\frac{\binom{n-k}{i-s}\sqrt{\binom{k}{s}\binom{k}{j-i+s}}}{\sqrt{\binom{n}{i}\binom{n}{j}}}
\ketbra{D_k^{s}}{D_k^{j-i+s}}=\sum_{0\leq i< j\leq n}a_{j,i}\sum_{s}\frac{\binom{n-k}{i-s}\sqrt{\binom{k}{s}\binom{k}{j-i+s}}}{\sqrt{\binom{n}{i}\binom{n}{j}}}
\ketbra{D_k^{j-i+s}}{D_k^{s}}=\textbf{0},
\end{equation}
and 
\begin{equation}\label{eq:dig_equal}
\sum_{i=0}^na_{i,i}\sum_{s}\frac{\binom{k}{s}\binom{n-k}{i-s}}{\binom{n}{i}}\ketbra{D_k^s}{D_k^s}=\sum_{i=0}^k\lambda_i\sum_{s}\frac{\binom{k}{s}\binom{n-k}{i-s}}{\binom{n}{i}}\ketbra{D_k^s}{D_k^s}.
\end{equation}
Eq.~\eqref{eq:dig_equal} is equivalent to 
\begin{equation}
   \sum_{s}\sum_{i=0}^na_{i,i}\frac{\binom{k}{s}\binom{n-k}{i-s}}{\binom{n}{i}}\ketbra{D_k^s}{D_k^s}=\sum_{s}\sum_{i=0}^k\lambda_i\frac{\binom{k}{s}\binom{n-k}{i-s}}{\binom{n}{i}}\ketbra{D_k^s}{D_k^s}. 
\end{equation}
Then we have
\begin{equation}\label{eq:onhomogeneous_linear_equations}
  \sum_{i=0}^na_{i,i}\frac{\binom{k}{s}\binom{n-k}{i-s}}{\binom{n}{i}}=\sum_{i=0}^k\lambda_i\frac{\binom{k}{s}\binom{n-k}{i-s}}{\binom{n}{i}},  \quad \forall \,0\leq s\leq k.  
\end{equation}
Since the  nonhomogeneous linear equations of Eq.~\eqref{eq:onhomogeneous_linear_equations} have only one solution $(a_{0,0},a_{1,1},\ldots,a_{k,k},a_{k+1,k+1}, \ldots, a_{n,n})=(\lambda_0,\lambda_1,\ldots,\lambda_k,0,\ldots,0)$,  we have $a_{i,j}=0$ for all $k+1\leq i,j \leq n$ ($A=(a_{i,j})_{0\leq i,j\leq n}$ is positive semidefinite).
Then Eq.~\eqref{eq:nondiag} is equivalent to 
\begin{equation}
\sum_{0\leq i<j\leq k}a_{i,j}\ketbra{D_n^i}{D_n^j}_{[k]}=\sum_{0\leq i<j\leq k}a_{j,i}\ketbra{D_n^j}{D_n^i}_{[k]}=\textbf{0}.   
\end{equation}
By Lemma~\ref{lemma:nondiag}, $a_{i,j}=0$ for all $0\leq i<j\leq k$. Since $a_{j,i}=a_{i,j}^{*}$, we have $a_{i,j}=0$ for all $0\leq i\neq j\leq k$.  

Above all, we have 
\begin{equation}
\sigma=\sum_{i=0}^k\lambda_i\ket{D_n^i}\bra{D_n^i}=\rho,
\end{equation}
and $\cC(\rho,\cS_k)$ contains only $\rho$. Thus, $L(\rho)\leq k$.
\end{proof}
\vspace{0.4cm}

\begin{lemma}\label{lemma:diagnoal_SDL_n}
Let $\rho=\sum_{i=0}^n\lambda_i\ket{D_n^i}\bra{D_n^i}$,  
then $L(\rho)=n$ if and only if one of the  three conditions holds:
\begin{enumerate}[(i)]
\item $\lambda_0\lambda_n\neq 0$;  
\item $\lambda_i\neq 0$ for all odd $i$;  \item $\lambda_i\neq 0$ for all even $i$. 
\end{enumerate}
\end{lemma}
\begin{proof}
    Sufficiency. There are three cases. 
    \begin{enumerate}[(i)]
        \item 
   If $\lambda_0\lambda_n\neq 0$, then there exists $0\neq a\in \bbR $ such that 
    $\lambda_0\lambda_n\geq a^2$. Let $\sigma=\rho +a\ketbra{D_n^0}{D_n^n}+a\ketbra{D_n^n}{D_n^0}$, then we can verify that ${\sf Tr}(\sigma)=1$ and $\sigma\geq 0$. Then $\rho\neq \sigma\in \cC(\rho, \cS_{n-1})$. Thus, we have $L(\rho)=n$.

       \item   Assume $\lambda_i\neq 0$ for all odd $i$. Let $k=n-1$ in Lemma~\ref{lemma:nonhom_general} and Lemma~\ref{lemma:nonhom_special}, then the general solution of the nonhomogeneous linear equations in   Eq.~\eqref{eq:nonhomo_general} is
\begin{equation}
   \begin{aligned}
a_{r,r}&=(-1)^{n-r}\binom{n-1}{r}\frac{n}{n-r}s_n+\lambda_r,  \quad \forall \,0\leq r\leq n-1,\\
a_{n,n}&=(-1)^{n-n}s_n+\lambda_n,
\end{aligned}
\end{equation}
where $s_n\in \mathbb{R}$ and $a_{r,r}\geq 0$ for all $0\leq r\leq n$.

\begin{enumerate}[1)]

\item  If $n$ is even, then $a_{r,r}=-\binom{n-1}{r}\frac{n}{n-r}s_n+\lambda_r$ for all odd $0\leq r\leq n-1$, and $a_{r,r}=\binom{n-1}{r}\frac{n}{n-r}s_n+\lambda_r$ for all even $0\leq r\leq n-1$. As long as $s_n>0$ is small enough, we can always find a solution $(a_{0,0},a_{1,1},\ldots,a_{n,n})\neq (\lambda_0,\lambda_1,\ldots,\lambda_n)$ with $a_{r,r}\geq 0$ for all $0\leq r\leq n$. Thus, according to Lemma~\ref{lemma:nonhom_general}, we have $L(\rho)=n$.   

\item If $n$ is odd, then $a_{r,r}=\binom{n-1}{r}\frac{n}{n-i}s_n+\lambda_r$ for all odd $0\leq r\leq n-1$, and $a_{r,r}=-\binom{n-1}{r}\frac{n}{n-r}s_n+\lambda_r$ for all even $0\leq r\leq n$. As long as $(-s_n)>0$ is small enough, we can always find a solution $(a_{0,0},a_{1,1},\ldots,a_{n,n})\neq (\lambda_0,\lambda_1,\ldots,\lambda_n)$ with $a_{r,r}\geq 0$ for all $0\leq r\leq n$. Thus, according to Lemma~\ref{lemma:nonhom_general}, we have $L(\rho)=n$.

\end{enumerate}
\item Assume $\lambda_i\neq 0$ for all even $i$. For the same disscusion as (ii), we have $L(\rho)=n$. 
 \end{enumerate}

Necessity. We prove it by contradiction. There are two cases. 

\begin{enumerate}[(1)]
    \item Assume $\lambda_n=0$, and there  exists an odd $i$ and an even $j$, such that $\lambda_i=0$ and $\lambda_j=0$. In this case, $\rho=\sum_{i=0}^{n-1}\lambda_i\ketbra{D_n^i}{D_n^i}$.
    Let $k=n-1$ in Lemma~\ref{lemma:nonhom_special} and Lemma~\ref{lem:nonhom}, then 
 the general solution of the nonhomogeneous linear equations in   Eq.~\eqref{eq:noneqlem} is
\begin{equation}
   \begin{aligned}
a_{r,r}&=(-1)^{n-r}\binom{n-1}{r}\frac{n}{n-r}s_n+\lambda_r,  \quad \forall \,0\leq r\leq n-1,\\
a_{n,n}&=(-1)^{n-n}s_n,
\end{aligned}
\end{equation}
where $s_n\in \mathbb{R}$ and $a_{r,r}\geq 0$ for all $0\leq r\leq n$.
Note that $a_{i,i}=(-1)^{n-i}xs_n$,  $a_{j,j}=(-1)^{n-j}ys_n$, where $x,y> 0$. Since $a_{i,i}, a_{j,j}\geq 0$, we obtain $s_n=0$, and $(a_{0,0},a_{1,1},\ldots, a_{n-1,n-1},a_{n,n})=(\lambda_0,\lambda_1,\ldots,\lambda_{n-1},0)$. According to Lemma~\ref{lem:nonhom}, we have $L(\rho)\leq n-1$. 

\item   Assume $\lambda_0=0$, and there  exists an odd $i$ and an even $j$, such that $\lambda_i=0$ and $\lambda_j=0$.  In this case, $\rho=\sum_{i=1}^{n}\lambda_i\ketbra{D_n^i}{D_n^i}$. Note that $\rho'=(\otimes_{i=1}^n X)\rho(\otimes_{i=1}^n X)=\sum_{i=0}^{n-1}\lambda_{n-i}\ketbra{D_n^{i}}{D_n^{i}}$, where $X$ is the Pauli matrix. According to (1), we have $L(\rho')\leq n-1$.
By Lemma~\ref{lemma:L_LU}, we have $L(\rho)=L(\rho')\leq n-1$.
\end{enumerate}
\end{proof}

Note that when $n$ is even,  then the three conditions in    Lemma~\ref{lemma:diagnoal_SDL_n} are reduced to the first two conditions.

\begin{lemma}\label{lem:mixturedicke}
Let $\rho=\sum_{i=0}^k\lambda_i\ket{D_n^i}\bra{D_n^i}$, where $1\leq k\leq n$. Assume there exists $i_*$ ($0\leq i_*\leq k$) such that $k-i_*$ is even and $\lambda_{i_*}=0$, 
  then $L(\rho)\leq k$.
\end{lemma}
\begin{proof}
According to the proof of Lemma~\ref{lemma:nonhom_special}, we know that
the general solution of the nonhomogeneous linear equations in Eq.~\eqref{eq:noneqlem} is 
\begin{equation}\label{eq:solutionall}
\begin{aligned}
a_{r,r}&=\sum_{i=k+1}^n(-1)^{k-r+1}\binom{i}{k}\binom{k}{r}\frac{i-k}{i-r}s_i+\lambda_r,  \quad \forall \,0\leq r\leq k,\\
a_{r,r}&=s_r, \quad \forall \,k+1\leq r\leq n,
\end{aligned}
\end{equation}
where $s_r\in \mathbb{R}$ for all $k+1\leq r\leq n$,  and $a_{r,r}\geq 0$ for all $0\leq r\leq n$.

If there exists $i_*$ ($0\leq i_*\leq k$) such that $k-i_*$ is even and $\lambda_{i_*}=0$,
then $a_{i_*,i_*}=-\sum_{i=k+1}^n\binom{i}{k}\binom{k}{i_*}\frac{i-k}{i-i_*}s_i=-\sum_{i=k+1}^n\binom{i}{k}\binom{k}{i_*}\frac{i-k}{i-i_*}a_{i,i}\geq 0$,  which implies $a_{r,r}=s_r=0$ for all $k+1\leq r\leq n$. Thus $(a_{0,0},a_{1,1},\ldots,a_{k,k},a_{k+1,k+1}, \ldots, a_{n,n})=(\lambda_0,\lambda_1,\ldots,\lambda_k,0,\ldots,0)$. According to Lemma~\ref{lem:nonhom}, we have $L(\rho)\leq k$.
\end{proof}
\vspace{0.4cm}

\begin{proposition}  Let  $\rho$ be a symmetric state of the diagonal  form   $\rho=\sum_{i=0}^k\lambda_i\ket{D_n^i}\bra{D_n^i}$ for some integer $k\le n$, where    $\lambda_i$ is nonzero for all but one value of $i$, namely $\lambda_i\not  =  0,  \, \forall \, i\in  \{0,\dots, k\} \setminus \{i_*\}$ for some $i_*\in  \{0,\dots,  k\}$.
  Then,  $L(\rho)=k   + \min  \{[(k-i_*)\mod 2],  |n-k|\}$.  
  \end{proposition}
\begin{proof}





\begin{enumerate}[(A)]

 \item Since $L(\ket{D_n^0})=1$ and $L(\ket{D_n^1})=2$, Proposition~\ref{pro:symmetric_SDL} holds for $k=1$.


 \item   Assume $k=2$, then $\lambda_i=0$ for $3\leq i\leq n$. 

\begin{enumerate}[1)]
    \item If $2-i_*$ is even, then we have $L(\rho)\leq 2$ by Lemma~\ref{lem:mixturedicke}. According to Lemma~\ref{lemma:Lrho1}, we know that $L(\rho)\geq 2$. Thus $L(\rho)= 2$.

 \item If $2-i_*$ is odd,  then we have  $\lambda_1=0$, $\lambda_0, \lambda_2>0$. Assume $k=2$ in Eq.~\eqref{eq:solutionall}, then 
\begin{equation}
\begin{aligned}
    a_{0,0}&=-\sum_{i=3}^n\frac{(i-1)(i-2)}{2}s_i+\lambda_0, \\
    a_{1,1}&=\sum_{i=3}^ni(i-2)s_i,\\
    a_{2,2}&=-\sum_{i=3}^n\frac{i(i-1)}{2}s_i+\lambda_2,\\
    a_{r,r}&=s_r,  \quad \forall \,3\leq r\leq n,
\end{aligned}
\end{equation}
where $s_r\in \mathbb{R}$ for all $3\leq r\leq n$,  and $a_{r,r}\geq 0$ for all  $0\leq r\leq n$.  As long as $s_r>0$ is small enough for all $3\leq r\leq n$, we can always find a solution  $(a_{0,0},a_{1,1},a_{2,2},a_{3,3},\ldots,a_{n,n})\neq (\lambda_0,0,\lambda_2,0,\ldots,0)$ with $a_{r,r}\geq 0$ for all $0\leq r\leq n$.  According to Lemma~\ref{lemma:nonhom_general}, we have $L(\rho)\geq 3$. 
For example, if $n=4$, $\lambda_0=\frac{1}{2}$, and $\lambda_2=\frac{1}{2}$, then we have find a solution $(a_{0,0},a_{1,1},a_{2,2},a_{3,3},a_{4,4})=(\frac{8}{24},\frac{11}{24},\frac{3}{24},\frac{1}{24},\frac{1}{24})$, which means that $\frac{8}{24}\ket{D_4^0}\bra{D_4^0}+\frac{11}{24}\ket{D_4^1}\bra{D_4^1}+\frac{3}{24}\ket{D_4^2}\bra{D_4^2}+\frac{1}{24}\ket{D_4^3}\bra{D_4^3}+\frac{1}{24}\ket{D_4^4}\bra{D_4^4} \in \cC(\frac{1}{2}\ket{D_4^0}\bra{D_4^0}+\frac{1}{2}\ket{D_4^2}\bra{D_4^2}, \cS_2)$.
Let $k=3$ in Lemma~\ref{lem:mixturedicke}, then $3-1$ is even, we have $L(\rho)\leq 3$, which implies  $L(\rho)=3$.
 \end{enumerate}



 \item   Assume $k=3$, then $\lambda_i=0$ for $4\leq i\leq n$.  
 \begin{enumerate}[1)]
     \item 
If $3-i_*$ is even (i.e., $i_*=1$ or $i_*=3$), then we have $L(\rho)\leq 3$ by Lemma~\ref{lem:mixturedicke}. 
 In order to show  $L(\rho)\geq 3$, we need to assume $k=2$ in Eq.~\eqref{eq:solutionall}, and we have   
\begin{equation}
\begin{aligned}
    a_{0,0}&=-\sum_{i=3}^n\frac{(i-1)(i-2)}{2}s_i+\lambda_0, \\
    a_{1,1}&=\sum_{i=3}^ni(i-2)s_i+\lambda_1,\\
    a_{2,2}&=-\sum_{i=3}^n\frac{i(i-1)}{2}s_i+\lambda_2,\\
    a_{3,3}&=s_3+\lambda_3,\\
    a_{r,r}&=s_r,\quad \forall \,4\leq r\leq n,
\end{aligned}
\end{equation}
where $s_r\in \mathbb{R}$ for all $3\leq r\leq n$,  and $a_{r,r}\geq 0$ for all  $0\leq r\leq n$. If $\lambda_1=0$, and $\lambda_0,\lambda_2,\lambda_3> 0$, 
then as long as $s_r>0$ is small enough for all $3\leq r\leq n$,    we can always find a solution  $(a_{0,0},a_{1,1},a_{2,2},a_{3,3},a_{4,4}\ldots,a_{n,n})\neq (\lambda_0,0,\lambda_2,\lambda_3,0\ldots,0)$ with $a_{r,r}\geq 0$ for all $0\leq r\leq n$. 
By Lemma~\ref{lemma:nonhom_general}, we have $L(\rho)\geq 3$. For example, if $n=4$, $\lambda_0=\lambda_2=\lambda_3=\frac{1}{3}$, and  then we have find a solution $(a_{0,0},a_{1,1},a_{2,2},a_{3,3},a_{4,4})=(\frac{12}{48},\frac{11}{48},\frac{7}{48},\frac{17}{48},\frac{1}{48})$, which means that $\frac{12}{48}\ket{D_4^0}\bra{D_4^0}+\frac{11}{48}\ket{D_4^1}\bra{D_4^1}+\frac{7}{48}\ket{D_4^2}\bra{D_4^2}+\frac{17}{48}\ket{D_4^3}\bra{D_4^3}+\frac{1}{48}\ket{D_4^4}\bra{D_4^4} \in \cC(\frac{1}{3}\ket{D_4^0}\bra{D_4^0}+\frac{1}{3}\ket{D_4^2}\bra{D_4^2}+\frac{1}{3}\ket{D_4^3}\bra{D_4^3}, \cS_2)$.
If $\lambda_3=0$, and $\lambda_0,\lambda_1,\lambda_2> 0$ then we can also obtain $L(\rho)\geq 3$ for the same discussion as above. Thus $L(\rho)=3$.

\item If $3-i_*$ is odd (i.e., $i_*=0$ or $i_*=2$), then we assume $k=3$ in  Eq.~\eqref{eq:solutionall}, and we have  
 \begin{equation}
 \begin{aligned}
     a_{0,0}&=\sum_{i=4}^n\frac{(i-1)(i-2)(i-3)}{3!}s_i+\lambda_0, \\
     a_{1,1}&=-\sum_{i=4}^n\frac{i(i-2)(i-3)}{2}s_i+\lambda_1,\\
     a_{2,2}&=\sum_{i=4}^n\frac{i(i-1)(i-3)}{2}s_i+\lambda_2,\\
     a_{3,3}&=-\sum_{i=4}^n\frac{i(i-1)(i-2)}{3!}s_i+\lambda_3,\\
     a_{r,r}&=s_r,  \quad \forall \,4\leq r\leq n,
 \end{aligned}
 \end{equation}
where $s_r\in \mathbb{R}$ for all $4\leq r\leq n$, and $a_{r,r}\geq 0$ for all  $0\leq r\leq n$. Similarly, as long as $s_r$ ($s_r>0$) is small enough for all $4\leq r\leq n$,    we can always find a solution  $(a_{0,0},a_{1,1},a_{2,2},a_{3,3},a_{4,4}\ldots,a_{n,n})\neq (\lambda_0,\lambda_1,\lambda_2,\lambda_3,0\ldots,0)$ with $a_{r,r}\geq 0$ for all $0\leq r\leq n$. 
By Lemma~\ref{lemma:nonhom_general}, we have $L(\rho)\geq 4$. 
Let $k=4$ in Lemma~\ref{lem:mixturedicke}, then $4-j$ is even, and we have $L(\rho)\leq 4$. Thus $L(\rho)=4$.
 \end{enumerate}

\item  Now, we consider the general $k$, where $2\leq k\leq n-1$. Note that $\lambda_i=0$ for $k+1 \leq i\leq n$. 

\begin{enumerate}[1)]
\item   If $k-i_*$ is even, then  we have $L(\rho)\leq k$ by Lemma~\ref{lem:mixturedicke}. In order to show $L(\rho)\geq k$, we need to replace 
 $k$ with $k-1$ in  Eq.~\eqref{eq:solutionall}, and we have  
\begin{equation}
\begin{aligned}
a_{r,r}&=\sum_{i=k}^n(-1)^{k-r}\binom{i}{k-1}\binom{k-1}{r}\frac{i-k+1}{i-r}s_i+\lambda_r,  \quad \forall \,0\leq r\leq k-1,\\
a_{k,k}&=s_k+\lambda_k,\\
a_{r,r}&=s_r, \quad \forall \,k+1\leq r\leq n,
\end{aligned}
\end{equation}
where $s_r\in \mathbb{R}$ for all $k\leq r\leq n$, and $a_{r,r}\geq 0$ for all  $0\leq r\leq n$.  If $k-i_*$ is even and $i_*\neq k$, then  we have $a_{i_*,i_*}=\sum_{i=k}^n(-1)^{k-i_*}\binom{i}{k-1}\binom{k-1}{i_*}\frac{i-k+1}{i-i_*}s_i+\lambda_{i_*}=
\sum_{i=k}^n\binom{i}{k-1}\binom{k-1}{i_*}\frac{i-k+1}{i-i_*}s_i$. As long as $s_r>0$ is small enough for all $k\leq r\leq n$,    we can always find a solution  $(a_{0,0},a_{1,1},\ldots,a_{k,k},a_{k+1,k+1}\ldots, a_{n,n})\neq (\lambda_0,\lambda_1,\ldots,\lambda_k,0\ldots,0)$ with $a_{r,r}\geq 0$ for all $0\leq r\leq n$. By Lemma~\ref{lemma:nonhom_general}, we have $L(\rho)\geq k$. If $i_*=k$, we also have $L(\rho)\geq k$ for the same discussion as above. Thus $L(\rho)=k$.

\item  If $k-i_*$ is odd, then   Eq.~\eqref{eq:solutionall} is  
\begin{equation}
\begin{aligned}
a_{r,r}&=\sum_{i=k+1}^n(-1)^{k-r+1}\binom{i}{k}\binom{k}{r}\frac{i-k}{i-r}s_i+\lambda_r,  \quad \forall \,0\leq r\leq k,\\
a_{r,r}&=s_r, \quad \forall \,k+1\leq r\leq n,
\end{aligned}
\end{equation}
where $s_r\in \mathbb{R}$ for all $k+1\leq r\leq n$, and $a_{r,r}\geq 0$ for all  $0\leq r\leq n$. Note that $a_{i_*,i_*}=\sum_{i=k+1}^n(-1)^{k-i_*+1}\binom{i}{k}\binom{k}{i_*}\frac{i-k}{i-i_*}s_i+\lambda_{i_*}=\sum_{i=k+1}^n\binom{i}{k}\binom{k}{i_*}\frac{i-k}{i-i_*}s_i$.  As long as $s_r>0$ is small enough for all $k+1\leq r\leq n$,    we can always find a solution  $(a_{0,0},a_{1,1},\ldots,a_{k,k},a_{k+1,k+1},\ldots, a_{n,n})\neq (\lambda_0,\lambda_1,\ldots,\lambda_k,0\ldots,0)$ with $a_{r,r}\geq 0$ for all $0\leq r\leq n$. By Lemma~\ref{lemma:nonhom_general}, we have $L(\rho)\geq k+1$. If we replace $k$ with $k+1$ in Lemma~\ref{lem:mixturedicke}, then $k+1-i_*$ is even, and we have $L(\rho)\leq k+1$. Thus $L(\rho)=k+1$.
 \end{enumerate}
 \item According to Lemma~\ref{lemma:diagnoal_SDL_n}, Proposition~\ref{pro:symmetric_SDL} holds for $k=n$.
\end{enumerate}
This completes the proof. 
\end{proof}  
\vspace{0.4cm}

We list some examples. 
 Let $k=2$ in Proposition~\ref{pro:symmetric_SDL}, then $L(\lambda_0\ketbra{D_n^0}{D_n^0}+\lambda_1\ketbra{D_n^1}{D_n^1})=L(\lambda_1\ketbra{D_n^1}{D_n^1}+\lambda_2\ketbra{D_n^2}{D_n^2})=2$ when $\lambda_0\lambda_1\lambda_2\neq 0$;   $L(\lambda_0\ketbra{D_n^0}{D_n^0}+\lambda_2\ketbra{D_n^2}{D_n^2})=3$ when $\lambda_0\lambda_2\neq 0$.

\section{VI. \ The proof of Proposition~3}
\begin{proposition}
For an $n$-qubit entangled symmetric state $\rho$, a  collection $\cS$ determines $\rho$ (detects $\rho$'s GME) if and only if   the hypergraph $G=  ([n],\cS)$ is connected and $|S| \ge L(\rho)$ ($|S| \ge l(\rho)$)  $\exists \, S\in \cS$. 
\end{proposition}
\begin{proof}
  According to Lemma~\ref{lemma:Lrho1}, $L(\rho)\geq 2$ for an $n$-qubit entangled symmetric state $\rho$.

Sufficiency.   If $G=([n], \cS)$ is connected, then   every $\sigma\in \cC(\rho, \cS)$ is a symmetric state according to Lemma~\ref{lemma: connected_symmetrc}. Since there exists a subset $S\in \cS$ with $|S| \ge L(\rho)$ (respectively, $|S| \ge l(\rho)$), we have   $\sigma_S=\rho_S$. According to Lemma~\ref{lem:rdm},  we obtain that $\sigma_{S'}=\sigma_{[L(\rho)]}=\rho_{[L(\rho)]}=\rho_{S'}$  for every $S'\in \cS_{L(\rho)}$ (respectively, $\sigma_{S'}=\sigma_{[l(\rho)]}=\rho_{[l(\rho)]}=\rho_{S'}$  for every $S'\in \cS_{l(\rho)}$), which implies $\sigma \in \cC(\rho, \cS_{L(\rho)})$ (respectively,  $\sigma \in \cC(\rho, \cS_{l(\rho)})$. Since $\cC(\rho, \cS_{L(\rho)})$ (respectively, $\cC(\rho, \cS_{l(\rho)})$) contains only $\rho$ (respectively,  only genuinely entangled states), we obtain that $\sigma=\rho$ (respectively,  $\sigma$ is genuinely entangled). Thus $\cS$  determines $\rho$ (respectively, $\cS$ detects $\rho$'s GME).

Necessity. According to the definition of SDL ((respectively, EDL), there must exist a subset $S\in \cS$ such that $|S| \ge L(\rho)$ (respectively,
$|S| \ge l(\rho)$). Assume that $G=([n], \cS)$ is not connected, then the vertex set $[n]$ can be partitioned into two non-empty disjoint sets $X$ and $Y$ such that no hyperedge connects a vertex in $X$ to a vertex in $Y$. We have  $\rho_{X}\otimes \rho_{Y}\in \cC(\rho, \cS)$. Since $\rho$ is entangled,  $\rho_{X}\otimes \rho_{Y}\neq \rho$. This contradicts that $\cS$  determines $\rho$ (respectively, $\cS$ detects $\rho$'s GME). 
\end{proof}
\vspace{0.4cm}

  Next, we consider  the minimum number of $L(\rho)$-body marginals (respectively, $l(\rho)$-body marginals)  needed to determine  $\rho$ (respectively, detect $\rho$'s GME), i.e., $M(\rho)$ (respectively, $m(\rho)$).

\begin{figure*}[b]
		\centering		\includegraphics[scale=0.5]{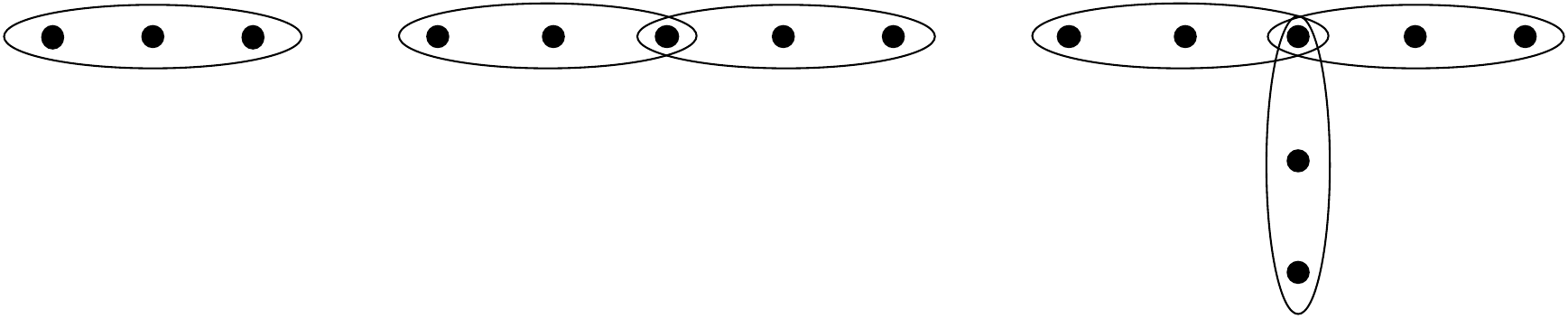}
		\caption{ Three connected $3$-uniform  hypergraphs. The left hypergraph has one hyperedge; the center hypergraph has two hyperedges; and the right hypergraph has three hyperedges.  } \label{fig:hypergraph}
\end{figure*}

\begin{lemma}\label{lemma: minimum_edges}
    For an $n$-qubit entangled symmetric state $\rho$, $M(\rho)   =\lceil  (n-1)/(L(\rho)-1) \rceil $  ($m(\rho)   =\lceil  (n-1)/(l(\rho)-1) \rceil$). 
\end{lemma}
\begin{proof} 
 According to Proposition~\ref{prop:graphtheory}, we need to show that for a  connected $k$-uniform hypergraph with $n$ vertices, the minimum number of hyperedges is $\fc{n-1}{k-1}$.

For a connected $k$-uniform  hypergraph, 1 hyperedge covers $k$ vertices; 2 hyperedges cover at most $2k-1$ vertices;  3 hyperedges cover at most $3k-2$ vertices (see Fig.~\ref{fig:hypergraph}).  Then $s$ hyperedges cover at most $sk-s+1$ vertices. For a connected $k$-uniform  hypergraph with $n$ vertices and $s$ hyperedges, it must have   $sk-s+1\geq n$, i.e. 
     $s\geq \fc{n-1}{k-1}$. Thus, the minimum number of hyperedges for a connected $k$-uniform  hypergraph with $n$ vertices is $\fc{n-1}{k-1}$.
\end{proof}
\vspace{0.4cm}



\section{VII. \ The proof of Proposition 4}
\begin{proposition}
   For an $n$-qubit  entangled symmetric state $\rho$, if the hypergraph $G=  ([n],\cS)$ is connected and $|S| \ge l(\rho)$  $\exists \, S\in \cS$, then  the set of marginals $\{\rho_S\mid S\in \cS \}$ exhibits entanglement transitivity in any $T$ with $|T|\geq l(\rho)$.   
\end{proposition}
\begin{proof}
If the hypergraph $G=([n],\cS)$ is connected, then every $\sigma\in \cC(\rho, \cS)$ is a symmetric state according to Lemma~\ref{lemma: connected_symmetrc}. Since there exists a subset $S\in \cS$ with $|S| \ge l(\rho)$, we have $\sigma_S=\rho_S$.  According to Proposition~\ref{prop:symmetric}, we know that $\rho_{[l(\rho)]}$ is entangled. By using Lemma~\ref{lem:rdm}, we obtain that $\sigma_{T}=\sigma_{[l(\rho)]}=\rho_{[l(\rho)]}=\rho_T$ is entangled for every  $T\in \cS_{l(\rho)}$. For every $T\subseteq [n]$ with $|T|> l(\rho)$, there exists a subset $S\in \cS_{l(\rho)}$ such that $S\subsetneq T$, then $\sigma_T$ must be entangled.
\end{proof}
\vspace{0.4cm}


Next, we consider two  examples. 

\begin{enumerate}
    \item Let 
$\rho=\lambda_0\ket{D_n^0}\bra{D_n^0}+\lambda_1\ket{D_n^1}\bra{D_n^1}$, where $n\geq 2$ and $\lambda_1\neq 0$.  We have shown that $l(\rho)=2$ in Sec.~\ref{sec:symmetric_EDL}.  If the simple graph $G=([n], \cS)$ is connected, then the set of marginals $\{\rho_S\mid S\in \cS \}$ exhibits entanglement transitivity in any $T$ with $|T|\geq 2$,   which recovers one of the main results in Ref.~\cite{tabia2022entanglement}. 

\item  Let $\rho=\lambda_0\ket{D_n^0}\bra{D_n^0}+\lambda_2\ket{D_n^2}\bra{D_n^2}$, where $n\geq 3$ and $0<\lambda_2\leq \frac{n}{2n-2}$. We have shown that $l(\rho)=3$ in Sec.~\ref{sec:symmetric_EDL}.
If the   hypergraph $G=([n], \cS)$ is connected and $|S|\geq 3$ $\exists \, S\in \cS$, then the set of marginals $\{\rho_S\mid S\in \cS \}$ exhibits entanglement transitivity in any $T$ with $|T|\geq 3$.
\end{enumerate}

\section{VIII. \ SDP for EDL and SDL}

 An $n$-qubit observable  $W$ is fully decomposable witness \cite{jungnitsch2011taming}  if for every bipartition $S|\overline{S}$, there exist positive semidefinite operators $P_S$ and $Q_S$ such that $W=P_S+Q_S^{T_S}$, where $T_S$ is the partial transpose over $\cH_S$.  This definition implies that the expectation value of $W$ is non-negative for every biseparable state. Hence, a negative value of ${\sf Tr} [W\rho]$ implies  that $\rho$ is genuinely entangled. We denote by $H^S$  an Hermitian operator on $\cH_S$ and by $I^{\overline{S}}$ an identity operator on  $\cH_{\overline S}$.
 Given a state $\rho$ and a collection $\cS$  of all subsets of $[n]$, the SDP is as follows:

\begin{equation}\label{eq:sdp1}
\begin{aligned}
\alpha  (\rho,  \cS) := \min \quad &{\sf Tr}(W\rho)\\
\text{s.t.} \quad &{\sf Tr}(W)=1, \\
&W=\sum_{S\in \cS} H^{S}\otimes \bbI^{\overline{S}},\\
&\text{$W$ is fully decomposable.}
\end{aligned}
\end{equation}

If $\alpha  (\rho,  \cS_k)<0$, then there exists a $k$-body witness $W$ such that ${\sf Tr}(W\rho)<0$, which means that $l(\rho)\leq k$.
We can immediately obtain the follwing result:

\begin{proposition}
For every state $\rho$, one has    $l(\rho)   \le  \min\{ k~|~  \alpha  (\rho,  \cS_k)  <0\}$. 
\end{proposition}

We also observe that evaluation of $l(\rho)$ from the SDP is robust to small amounts of white noise: for a state of the form  $\rho_p  =(1-p)\, \rho    +  p  \, I/2^n$, we have $l(\rho_p)\leq k$ for all  $p\in [0,\frac{2^n\alpha (\rho,  \cS_k)}{2^n\alpha  (\rho ,\cS_k)-1})$ when $\alpha (\rho,  \cS_k)<0$. This because when $\alpha (\rho,  \cS_k)<0$, we obtain a $k$-body witness $W$ such that    ${\sf Tr}(W\rho)=\alpha (\rho,  \cS_k)<0$, which also implies that ${\sf Tr}(W\rho_p)<0$ for all  $p\in [0,\frac{2^n\alpha (\rho,  \cS_k)}{2^n\alpha  (\rho ,\cS_k)-1})$.

Eq.~\eqref{eq:sdp1} can by solved by using the convex optimization package in Matlab \cite{cvx}.
We list some examples:

\begin{enumerate}[1.]
    \item Let $\ket{\psi}=\frac{1}{\sqrt{2}}\ket{1000}+\frac{1}{\sqrt{3}}\ket{0100}+\frac{1}{\sqrt{12}}\ket{0010}+\frac{1}{\sqrt{24}}\ket{0001}+\frac{1}{\sqrt{24}}\ket{1111}$.  We have $l(\ket{\psi})=2$ with  $p\in [0,0.1114)$.
    \item  Let $\rho=\frac{1}{3}\ketbra{D_4^0}{D_4^0}+\frac{1}{3}\ketbra{D_4^2}{D_4^2}+\frac{1}{3}\ketbra{1000}{1000}$. We have $l(\rho)\leq 3$ with $p\in [0, 0.1207)$.
    \item  Let $\rho=\frac{1}{2}\ketbra{D_3^1}{D_3^1}+\frac{1}{2}\ketbra{D_3^2}{D_3^2}$. According to Proposition~\ref{prop:symmetric}, we know that $\{\{1,2\},\{2,3\}\}$ detects $\rho$'s GME. By solving the SDP, we obtain the witness 
$W=0.125III-0.0556(XXI+IXX+YYI+IYY)-0.0139(ZZI+IZZ)$ with $p\in [0,0.1)$.  Thus, this method is useful for detecting GME with minimum observable length experimentally. 
\end{enumerate}



Ref.~\cite{tabia2022entanglement} gives a SDP to certify whether $\cS$ determines a pure state $\ket{\psi}$:

\begin{equation}\label{eq:sdp_pure}
\begin{aligned}
\alpha  (\ket{\psi},  \cS) := \min \quad &\bra{\psi}\rho\ket{\psi}\\
\text{s.t.} \quad &{\sf Tr}(\rho)=1, \\
&\rho\geq 0,\\
&\rho_S=\ketbra{\psi}{\psi}_S, \forall \, S\in \cS.
\end{aligned}
\end{equation}

According to this SDP, $\cS$ determines $\ket{\psi}$ if and only if $\alpha  (\ket{\psi},  \cS)=1$. Thus $L(\ket{\psi})  =  \min\{ k~|~  \alpha  (\ket{\psi},  \cS_k) =1\}$.
For example,   we reconsider the state $\ket{\psi}=\frac{1}{\sqrt{2}}\ket{1000}+\frac{1}{\sqrt{3}}\ket{0100}+\frac{1}{\sqrt{12}}\ket{0010}+\frac{1}{\sqrt{24}}\ket{0001}+\frac{1}{\sqrt{24}}\ket{1111}$,  then we have $L(\ket{\psi})=3$ from Eq.~\eqref{eq:sdp_pure}, while we have $l(\ket{\psi})=2$ from Eq.~\eqref{eq:sdp1}.
Note that, pure state  marginal problem can be also related to the separability problem \cite{yu2021complete}.

\section{IX. \ The maximum gap between EDL  and SDL.}

 Ref.~\cite{miklin2016multiparticle} gives a $4$-qubit genuinely  entangled pure state with non-zero gap between the EDL and SDL.
Let us recall this example of Ref.~\cite{miklin2016multiparticle}, 
\begin{equation}
    \ket{\Psi(\varphi)}=\frac{1}{\sqrt{2}}\ket{{\ghz}_4}+\frac{1}{\sqrt{2}}\ket{\widetilde{D_4^2}},
\end{equation}
with 
\begin{equation}
\begin{aligned}
     \ket{{\ghz}_4}&=\frac{1}{\sqrt{2}}(\ket{0000}+\ket{1111}), \\
     \ket{\widetilde{D_4^2}}&=\frac{1}{\sqrt{6}}(\ket{0011}+\ket{0101}+\ket{0110}+e^{i\varphi}\ket{1001}+e^{i\varphi}\ket{1010}+e^{-i\varphi}\ket{1100}),
\end{aligned}
\end{equation}
and $\varphi=\text{arccoss}(-\frac{1}{3})$.
Ref.~\cite{miklin2016multiparticle} shows that $\cC(\ket{\Psi(\varphi)}, \cS_2)$ contains only genuinely entangled states, and $\cC(\ket{\Psi(\varphi)}, \cS_2)$ contains a state $\ket{\Psi(-\varphi)}\neq \ket{\Psi(\varphi)}$, which means that the SDL is strictly greater than EDL for the $4$-qubit state $\ket{\Psi(\varphi)}$. 
Note that this example is found through SDP, which cannot be generalized to any $n$-qubit system. We will overcome this difficulty, and show that there is an $n$-qubit  genuinely entangled pure  (respectively, mixed) state with a maximum gap  between its EDL and SDL for general $n$.




\begin{proposition}
For an $n$-qubit genuinely entangled pure state $|\psi\rangle$, one has $\gap ( |\psi\rangle  )  \le \max\{0,n-3\} $.   The bound  is attained by  every  state of the form  $\ket{\psi}=\alpha  \ket{D_n^1}+  \beta \, \ket{D_n^n}$, where $\alpha$ and $\beta$ are two amplitudes satisfying the condition     $\frac{n^2-2n}{n^2-2n+1}<|\alpha|^2<1$.
\end{proposition}
\begin{proof}
 We only need to consider $n\geq 3$. Since $\ket{\psi}$ is not LU-equivalent to an $n$-qubit generalized $\ghz$ state $a\ket{D_n^0}+b\ket{D_n^n}$ with $ab\neq 0$, we have $L(\ket{\psi})\leq n-1$ \cite{walck2009only}. Let
   \begin{equation}
\sigma=|\alpha|^2\ketbra{D_n^1}{D_n^1}+|\beta|^2\ketbra{D_n^n}{D_n^n}.
 \end{equation}
It is easy to see that  $\sigma \in \cC(\ket{\psi}, \cS_{n-2})$. Since $\rank(\sigma)=2$,
we have $\sigma\neq \ketbra{\psi}{\psi}$, which implies $L(\ket{\psi})\geq n-1$. Thus $L(\ket{\psi})=n-1$.
Note that
\begin{equation}
    \ketbra{\psi}{\psi}_{[2]}=\frac{(n-2)|\alpha|^2}{n}\ketbra{D_2^0}{D_2^0}+\frac{2|\alpha|^2}{n}\ketbra{D_2^1}{D_2^1}+|\beta|^2\ketbra{D_2^2}{D_2^2}
\end{equation}
is a diagonal symmetric state. According to Eq.~\eqref{eq:hankel}, the Hankel matrix $M_0$ is as follow:
\begin{equation}
    M_0=\begin{pmatrix}
      \frac{(n-2)|\alpha|^2}{n} &   \frac{|\alpha|^2}{n}\\
      \frac{|\alpha|^2}{n}  & |\beta|^2
    \end{pmatrix}=\begin{pmatrix}
      \frac{(n-2)|\alpha|^2}{n} &   \frac{|\alpha|^2}{n}\\
      \frac{|\alpha|^2}{n}  & 1-|\alpha|^2
    \end{pmatrix}.
\end{equation}
If $\frac{n^2-2n}{n^2-2n+1}<|\alpha|^2<1$, then $M_0$ is not positive semidefinite.  According to Lemma~\ref{Lemma: hankel},  $\ketbra{\psi}{\psi}_{[2]}$ is NPT, and $\ketbra{\psi}{\psi}_{[2]}$ is entangled.  According to  Proposition~\ref{prop:symmetric}, we have $l(\ket{\psi})=2$. Thus  $\gap(\ket{\psi})= n-3$.
\end{proof}

\begin{proposition}
For an $n$-qubit genuinely entangled mixed state $\rho$,  one has   $\gap  (\rho)\le \max\{0, n-2\}$.    The bound is attained by  every  state of the form   
$\rho=\sum_{i=0}^{n}\lambda_i\ketbra{D_n^i}{D_n^i}$ where $(\lambda_i)_{i=0}^n$ are probabilities satisfying the conditions $\lambda_0\lambda_n\neq 0$ and $\sum_{i,j=0}^n \,  (n-i)j \,\lambda_i\lambda_j  [  (n-i-1)  (j-1)-i(n-j)] < 0$. 
\end{proposition}

\begin{proof}
 We only need to consider $n\geq 2$. If $\lambda_0\lambda_n\neq 0$, then $L(\rho)=n$ according to Lemma~\ref{lemma:diagnoal_SDL_n}.
According to Lemma~\ref{lem:diganol_23}, if   $ \left[\sum_{i=0}^{n}(n-i)(n-i-1)\lambda_i\right]\left[\sum_{i=0}^ni(i-1)\lambda_i\right]-\left[\sum_{i=0}^{n}i(n-i)\lambda_i\right]^2=\sum_{i,j=0}^n \,  (n-i)j \,\lambda_i\lambda_j  [  (n-i-1)  (j-1)-i(n-j)]< 0$, then $\rho_{[2]}$ is NPT, and $\rho_{[2]}$ is entangled. According to Proposition~\ref{prop:symmetric}, we have $l(\rho)=2$. Thus $\gap(\rho)= n-2$.
\end{proof}
\vspace{0.4cm}

For example, for pure states, when $\ket{\psi}=\sqrt{0.94}\ket{D_5^1}+\sqrt{0.06}\ket{D_5^5}$,
$\gap(\ket{\psi})=2$; and 
  when $\ket{\psi}=\sqrt{0.97}\ket{D_6^1}+\sqrt{0.03}\ket{D_6^6}$, $\gap(\ket{\psi})=3$. For mixed states,
when $\rho=\frac{1}{12}\ketbra{D_3^0}{D_3^0}+\frac{1}{2}\ketbra{D_3^1}{D_3^1}+\frac{1}{3}\ketbra{D_3^2}{D_3^2}+\frac{1}{12}\ketbra{D_3^3}{D_3^3}$, $\gap(\rho)=1$; and when $\rho=\frac{1}{24}\ketbra{D_4^0}{D_4^0}+\frac{1}{3}\ketbra{D_4^1}{D_4^1}+\frac{1}{2}\ketbra{D_4^2}{D_4^2}+\frac{1}{12}\ketbra{D_4^3}{D_4^3}+\frac{1}{24}\ketbra{D_4^4}{D_4^4}$, $\gap(\rho)=2$.

\end{sloppypar}

 \end{widetext}

\end{document}